%% file: ccp-camera-final.tex
\documentclass[a4paper,UKenglish]{lipics}
\usepackage{etex} %THIS IS NEEDED to avoid compiling error because we are using too many packages that need \newdim (including Samson's bit register pictures)

\usepackage{microtype}
\usepackage{amsthm,amsmath,amssymb,amsfonts,latexsym,stmaryrd}
\usepackage[PostScript=dvips]{diagrams}
\usepackage[all]{xy}
\usepackage{graphicx}
\usepackage{epsfig}
\usepackage{tikz,ifdraft}
\usetikzlibrary{arrows,positioning,matrix,backgrounds,calc,fit,decorations.pathreplacing} 
\usepackage{nicefrac}

\usepackage{eurosym}

\theoremstyle{plain}
\newtheorem{proposition}[theorem]{Proposition}

\newcommand{\beqa}{\begin{eqnarray*}}
\newcommand{\eeqa}{\end{eqnarray*}\par\noindent}

\renewcommand{\emph}[1]{\textbf{#1}}

\newcommand{\ket}[1]{{|} #1\rangle}
\newcommand{\ie}{i.e.~}

\newcommand{\Set}{\mathbf{Set}}

\newcommand{\vn}{\varnothing}

\newcommand{\IFF}{\; \Longleftrightarrow \;}

\newcommand{\da}{{\downarrow}}

\newcommand{\EE}{\mathcal{E}}

\newarrow{Eq}=====
\newarrow{mon}>--->
\newarrow{rel}--+->
\newarrow{inc}C--->
\newarrow{mto}|--->
\newarrow{LRTo}<--->
\newarrow{dash}....>

\diagramstyle[newarrowhead=vee,newarrowtail=vee]

\newcommand{\op}{\mathsf{op}}

\newcommand{\supp}{\mathsf{supp}}

\newcommand{\ua}{{\uparrow}}

\newcommand{\ZZ}{\mathbb{Z}}
\newcommand{\rmap}[2]{\rho^{#1}_{#2}}

\newcommand{\FZ}{F_{\ZZ}}
\newcommand{\FF}{\mathcal{F}}
\newcommand{\Cech}{\u{C}ech~}
\newcommand{\Cohom}[1]{\mbox{\textit{\u{H}}$^{#1}$}(\UU, \FF)}

\newcommand{\obst}{\gamma}

\newcommand{\SQT}{\!\sqrt{2}}

\newcommand{\rsqa}{\; \leadsto \;}

\def\Mcomma{\text{ ,}}
\def\Msemicolon{\text{ ;}}
\def\Mdot{\text{ .}}

\def\ZZ{\mathbb{Z}}

\def\cat#1{\mathbf{\mathsf{#1}}}

\renewcommand{\Set}      {\cat{Set}        }  % sets and functions
\newcommand{\AbGrp}    {\cat{AbGrp}      }  % abelian groups
\newcommand{\RMod}     {\text{$R$-}\cat{Mod}        }  
\def\op{\mathsf{op}}

\usepackage{colonequals}
\def\defeq{\colonequals} 
\def\setdef#1#2{\left\{#1 \;\middle|\; #2\right\}}             % { f(x) | p(x) }
\def\enset#1{\mathopen{ \{ }#1\mathclose{ \} }} % {a,b,...z}
\newcommand{\fdec}    [3]{#1\colon #2 \longrightarrow #3}
\newcommand{\fdef}    [3]{#1\!\coloncolon #2 \longmapsto     #3}
\newcommand{\fdecdef} [5]{#1\colon #2 \longrightarrow #3 \coloncolon #4 \longmapsto #5}

\newcommand{\natfdec} [3]{#1\colon #2 \Rightarrow #3}

\newcommand{\adjointfdec}[4]{(#1 \vdash #2)\colon \xymatrix{#3 \ar@<.5ex>[r]^{f_*} & \ar@<.5ex>[l]^{f^*} #4}}

\def\tuple#1{\mathopen{\langle}#1\mathclose{\rangle}}

\def\Forall#1{\forall{#1}\boldsymbol{.}\;}

\def\LambdaAbs#1{\lambda{#1}\boldsymbol{.}}
\def\implies{\Rightarrow}

\def\ps{\mathcal{P}}

\def\zero{\mathbf{0}}

\def\ker{\mathsf{ker}\, }
\def\coker{\mathsf{coker}\, }
\def\img{\mathsf{im}\, }

\renewcommand{\supp}{\mathsf{supp}\,}

\newcommand{\aff}{\mathsf{aff}\,}
\newcommand{\LSpan}{\mathsf{Span}\,}
\newcommand{\Aff}{\mathsf{Aff}\,}
\newcommand{\affN}{\mathsf{aff}}    %for use when not applied 
\newcommand{\lspanU}[1]{\mathsf{span}_{#1}\,}
\newcommand{\affU}[1]{\mathsf{aff}_{#1}\,}

\newcommand{\Th}{\mathbb{T}}
\newcommand{\Sol}{\mathbb{M}}

\newcommand{\Se}{\mathcal{S}}
\newcommand{\SAffe}{(\Aff \Se)}

\def\EventShf{\mathcal{E}}

\def\M{\mathcal{M}}

\def\AvN{\mathsf{AvN}}
\def\CSC{\mathsf{CSC}}
\def\CLC{\mathsf{CLC}}
\def\SC{\mathsf{SC}}
\def\LC{\mathsf{LC}}

\renewcommand{\Cech}{\v{C}ech}
\newcommand{\FR}{F_{R}}
\newcommand{\FRD}{F_{R'}}

\newcommand{\FU}{\FF_{\bar{U}}}

\def\Nerve{\mathcal{N}}
\newcommand{\bdmap}{\partial}
\newcommand{\cobd}{\delta}
\newcommand{\Ccoch}[1]{C^{#1}(\M, \FF)}
\newcommand{\Cocyc}[1]{Z^{#1}(\M, \FF)}
\newcommand{\Cobound}[1]{B^{#1}(\M, \FF)}
\renewcommand{\Cohom}[1]{\check{H}^{#1}(\M, \FF)}
\newcommand{\CcochRes}[2]{C^{#1}(\M, \FF|_{#2})}
\newcommand{\CocycRes}[2]{Z^{#1}(\M, \FF|_{#2})}

\newcommand{\CohomRes}[2]{\check{H}^{#1}(\M, \FF|_{#2})}
\newcommand{\CcochRel}[2]{C^{#1}(\M, \FF_{\tilde{#2}})}
\newcommand{\CocycRel}[2]{Z^{#1}(\M, \FF_{\tilde{#2}})}

\newcommand{\CohomRel}[2]{\check{H}^{#1}(\M, \FF_{\tilde{#2}})}

\newdir{ >}{{}*!/-7.5pt/@{>}}

\title{Contextuality, Cohomology and Paradox}

\author[1]{Samson Abramsky}
\author[1]{Rui Soares Barbosa}
\author[1]{Kohei Kishida}
\author[1,2]{Raymond Lal}
\author[1]{Shane Mansfield}
\affil[1]{Department of Computer Science, University of Oxford, Oxford, U.K.\\
\textup{\texttt{\{samson.abramsky\,|\,rui.soares.barbosa\,|\,kohei.kishida\,|\,shane.mansfield\}@cs.ox.ac.uk}}}
\affil[2]{Faculty of Philosophy, University of Cambridge \\
\textup{\texttt{rl335@cam.ac.uk}}}

\authorrunning{J.\,Q. Open and J.\,R. Access} %mandatory. First: Use abbreviated first/middle names. Second (only in severe cases): Use first author plus 'et. al.'

\Copyright{John Q. Open and Joan R. Access}%mandatory, please use full first names. LIPIcs license is "CC-BY";  http://creativecommons.org/licenses/by/3.0/

\subjclass{ddd}
% mandatory: Please provide 1-5 keywords
\serieslogo{}%please provide filename (without suffix)
\volumeinfo%(easychair interface)
  {Billy Editor and Bill Editors}% editors
  {2}% number of editors: 1, 2, ....
  {Conference title on which this volume is based on}% event
  {1}% volume
  {1}% issue
  {1}% starting page number
\EventShortName{}
\DOI{10.4230/LIPIcs.xxx.yyy.p}% to be completed by the volume editor
\begin{document}

\maketitle

\begin{abstract}
Contextuality is a key feature of quantum mechanics that provides an important  non-classical resource for quantum information and computation.
Abramsky and Brandenburger %
used sheaf theory to give a general treatment of contextuality in quantum theory
[New Journal of Physics 13 (2011) 113036].
However, contextual phenomena are found in other fields as well, for example database theory.
In this paper, we shall develop this unified view of contextuality.
We provide two main contributions: first, we expose a remarkable connection between contexuality and logical paradoxes; secondly, we show that an important class of contextuality arguments has a topological origin.
More specifically, we show that ``All-vs-Nothing'' proofs of contextuality are witnessed by cohomological obstructions.
\end{abstract}

\section{Introduction}\label{sec:intro}

Contextuality is one of the key characteristic features of quantum mechanics. It has been argued that it provides the ``magic'' ingredient enabling quantum computation \cite{howard:14}. There have been a number of recent experimental verifications that Nature does indeed exhibit this highly non-classical form of behaviour \cite{zu:12,zhang:13}.

The study of quantum contextuality has largely been carried out in a concrete, example-driven fashion, which makes it appear highly specific to quantum mechanics. Recent work by the present authors \cite{abramsky:11,abramsky:11a} and others \cite{cabello2010non} has exposed the general mathematical structure of contextuality, enabling more general and systematic results. It has also made apparent that contextuality is a general and indeed pervasive phenomenon, which can be found in many areas of \emph{classical} computation, such as databases \cite{abramsky:12c} and constraints \cite{abramsky:13b}. The work in \cite{abramsky:11} makes extensive use of methods developed within the logic and semantics of computation.

The key idea from \cite{abramsky:11} is to understand contextuality as arising where we have a family of data which is \emph{locally consistent, but globally inconsistent}.
This can be understood, and very effectively visualised (see Fig.~\ref{fig:bundle}) in topological terms: we have a base space of \emph{contexts} (typically sets of variables which can be jointly measured or observed), a space of data or observations fibred over this space, and a family of \emph{local sections} (typically valuations of the variables in the context) in these fibres. This data is consistent locally, but not globally: there is no \emph{global section} defined on all the variables which reconciles all the local data. In topological language, we can say that the space is ``twisted'', and hence provides an \emph{obstruction} to forming a global section.

This provides a unifying description of a number of phenomena which at first sight seem very different:
\begin{itemize}
\item \textit{Quantum contextuality.} The local data arises from performing measurements on compatible sets of observables.
The fact that there is no global section corresponds to a no-go result for a hidden-variable theory to explain the observable data.
\item \textit{Databases.} The local data are the relation tables of the database. The fact that there is no global section corresponds to the failure in general of the \emph{universal relation assumption} \cite{fagin:82,maier:84}.
\item \textit{Constraint satisfaction.} The local data corresponds to the constraints, defined on subsets of the variables. The fact that there is no global section corresponds to the non-existence of a solution for the CSP.
\end{itemize}

In the present paper, we shall develop this unified viewpoint to give a logical perspective on contextuality. In particular, we shall look at contextuality in relation to \emph{logical paradoxes}:
\begin{itemize}
\item We find a direct connection between the structure of quantum contextuality and classic semantic paradoxes such as ``Liar cycles'' \cite{cook:04,walicki:09}.
%There are also close connections to the study of circuits with feedback \cite{berry2000foundations}.
\item Conversely, contextuality offers a novel perspective on these paradoxes. Contradictory cycles of references give rise to exactly the form of local consistency and global inconsistency we find in contextuality.
\end{itemize}

\subsection*{Mathematical structure}
Sheaf theory \cite{maclane:92} provides the natural mathematical setting for our analysis, since it is directly concerned with the passage from local to global.
In this setting, it is furthermore natural to use \emph{sheaf cohomology} to characterise contextuality. Cohomology is one of the major tools of modern mathematics, which has until now largely been conspicuous by its  absence, both in theoretical computer science, and in quantum information.
The use of cohomology to characterise contextuality was initiated in \cite{abramsky:11a}. In the present paper, we take the cohomological approach considerably further, taking advantage of situations in which the outcomes of observations have an algebraic structure. This applies, for example, in the case of the standard Pauli spin observables, which have eigenvalues in $\ZZ_{2}$.

We study a strong form of contextuality arising from so-called ``All-vs-Nothing'' arguments \cite{mermin:90b}. We give a much more general formulation of such arguments than has appeared previously, in terms of local consistency and global inconsistency of systems of linear equations. We also show how an extensive class of examples of such arguments arises in the stabiliser fragment of quantum mechanics, which plays an important r\^ole in quantum error correction \cite{nielsen:00} and measurement-based quantum computation \cite{raussendorf:01}.

We then show how all such All-vs-Nothing arguments are witnessed by the cohomological obstruction to the extension of local sections to global ones previously studied in \cite{abramsky:11a}. This obstruction is characterised in more abstract terms than previously, using the connecting homomorphism of the long exact sequence. Our main theorem establishes a hierarchy of properties of probability models, relating their algebraic, logical and topological structures. 

For further details and development of the ideas, see the full version of the paper \cite{abramsky2015contextuality}.

%\subsubsection*{Outline}
%The further structure of the paper is as follows.
%In \S\ref{sec:manyfaces}, we show how  our  topological approach provides a unified description of many apparently different phenomena, linking contextuality, paradoxes and the ``logical twisting'' of spaces, leading to locally consistent data with obstructions to global consistency.
%This approach is then formalised in \S\ref{sec:sheafframework}.
%In \S\ref{sec:AvN}, we introduce All-vs-Nothing arguments and show how they arise concretely in stabiliser states. We then give our general formulation of these arguments as global logical inconsistencies.
%In \S\ref{sec:cohom}, we  develop the formal language of sheaf cohomology and apply it to provide cohomological witnesses of contextuality.
%Finally, in \S\ref{sec:cohomAvN}, we prove our main result, which relates the non-vanishing of the cohomology invariant to global logical inconsistency.

\section{The many faces of contextuality}\label{sec:manyfaces}

We begin with the following scenario, depicted in Fig.~\ref{fig:bitreg} (a).
Alice and Bob are agents positioned at nodes of a network. Alice can access local bit registers $a_1$ and $a_2$, while Bob can access local bit registers $b_1$, $b_2$. Alice can load one of her bit registers into a processing unit, and test whether it is $0$ or $1$. Bob can perform the same operations with respect to his bit registers.
They send the outcomes of these operations to a common target, which keeps a record of the joint outcomes.

\begin{figure}
\begin{center}
\begin{tabular}{cc}
\input{pictures/bitregs.tex}
\quad & \quad
\input{pictures/sbitregs.tex} \\
(a) \quad & \quad (b)
\end{tabular}
\end{center}
\caption{(a) Alice and Bob look at bits. (b) A source.}
\label{fig:bitreg}
\end{figure}
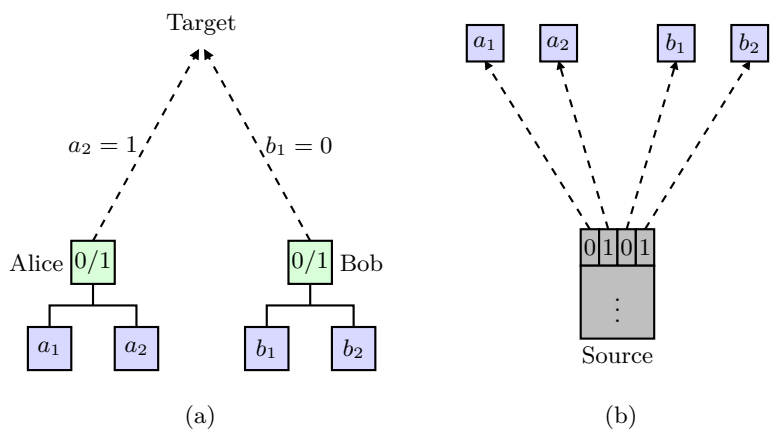

We now suppose that Alice and Bob perform repeated rounds of these operations. On different rounds, they may make different choices of which bit registers to access, and they may observe different outcomes for a given choice of register.
The target can compile statistics for this series of data, and infer probability distributions on the outcomes.

\subsection{Logical forms of contextuality}\label{ssec:bundle}

While contextuality can exhibit itself at the level of probability distributions (see \cite{abramsky:11,abramsky2015contextuality}), here we consider a stronger form of contextuality which exhibits itself at the level of the \emph{supports} of the distributions, highlighting a direct connection with logic.

Consider the tables in Fig.~\ref{fig:HardyPR},
\begin{figure}
\begin{center}
\begin{tabular}{cc}
\begin{tabular}{ll|ccccc}
A & B & $(0, 0)$ & $(1, 0)$ & $(0, 1)$ & $(1, 1)$  &  \\ \hline
$a_1$ & $b_1$ &  $1$ &  &  &  \\
$a_1$ & $b_2$ &   $0$ &  &  &  \\
$a_2$ & $b_1$ &  $0$ &  &  & \\
$a_2$ & $b_2$ &   &  &  & $0$ \\
\end{tabular}
\quad & \quad
\begin{tabular}{ll|ccccc}
A & B & $(0, 0)$ & $(1, 0)$ & $(0, 1)$ & $(1, 1)$  &  \\ \hline
$a_1$ & $b_1$ & $1$ & $0$ & $0$ & $1$ & \\
$a_1$ & $b_2$ & $1$ & $0$ & $0$ & $1$ & \\
$a_2$ & $b_1$ & $1$ & $0$ & $0$ & $1$ & \\
$a_2$ & $b_2$ & $0$ & $1$ & $1$ & $0$ & 
\end{tabular}
\end{tabular}
\end{center}
\caption{The Hardy paradox (left) and the PR box (right).}
\label{fig:HardyPR}
\end{figure}%
which depict the kind of scenario we have been considering. The entries are either $0$ or $1$. The idea is that a $1$ entry represents a positive probability. Thus we are distinguishing only between \emph{possible} (positive probability) and \emph{impossible} (zero probability). In other words, the rows correspond to the \emph{supports} of some (otherwise unspecified) probability distributions.
Note that only four entries of the Hardy table are filled in. Our claim is that just from these entries, referring only to the supports, we can deduce that there is no classical explanation for the behaviour recorded in the table. Moreover, this behaviour can be realised in quantum mechanics \cite{hardy:93}, yielding a stronger form of Bell's theorem \cite{bell:64}, due to Hardy \cite{hardy:93}.

\subsubsection*{What do ``observables'' observe?\nopunct}
Classically, we would take the view that physical observables directly reflect properties of the physical system we are observing. These are objective properties of the system, which are independent of our choice of which measurements to perform,
\ie of our \emph{measurement context}.
More precisely, this would say that for each possible state of the system, there is a function $\lambda$ which for each measurement $m$ specifies an outcome $\lambda(m)$, \emph{independently of which other measurements may be performed}.
This point of view is called \emph{non-contextuality}, and may seem self-evident.
However, this view is \emph{impossible to sustain} in the light of our \emph{actual observations of (micro)-physical reality}.

Consider once again the Hardy table depicted in Fig.~\ref{fig:HardyPR}. Suppose there is a function $\lambda$ which accounts for the possibility of Alice observing value $0$ for $a_1$ and Bob observing $0$ for $b_1$, as asserted by the entry in the top left position in the table. Then this function $\lambda$ must satisfy
$\lambda(a_1) = 0$, $\lambda(b_1) = 0$.
% \[ \lambda \coloncolon\, a_1 \mapsto 0, \quad b_1 \mapsto 0 \Mdot \]
Now consider the value of $\lambda$ at $b_2$. If $\lambda(b_2) = 0$, then this would imply that the event that $a_1$ has value $0$ and $b_2$ has value $0$ is possible. However, \emph{this is precluded} by the $0$ entry in the table for this event. The only other possibility is that $\lambda(b_2) = 1$. Reasoning similarly with respect to the joint values of $a_2$ and $b_2$, we conclude, using the bottom right entry in the table, that we must have $\lambda(a_2) = 0$. Thus the only possibility for $\lambda$ consistent with these entries is
% \[ \lambda \coloncolon\, a_1 \mapsto 0, \quad a_2 \mapsto 0, \quad b_1 \mapsto 0, \quad b_2 \mapsto 1 \Mdot \]
$\lambda \coloncolon\, a_1 \mapsto 0, a_2 \mapsto 0, b_1 \mapsto 0, b_2 \mapsto 1$.
But this would require the outcome $(0, 0)$ for measurements $(a_2,b_1)$ to be possible, and this is \emph{precluded} by the table.

%It is easily seen that this argument amounts to showing that the following formulas are not jointly satisfiable:
%\[ a_1 \wedge b_1, \quad \neg(a_1 \wedge b_2), \quad \neg(a_2 \wedge b_1), \quad a_2 \vee b_2 \Mdot \]
We are thus forced to conclude that the Hardy models are contextual. Moreover, we can say that they are contextual in a logical sense, stronger than the probabilistic form we saw with the Bell tables, since we only needed information about possibilities to infer the contextuality of this behaviour.

\subsubsection*{Strong contextuality}

Logical contextuality as exhibited by the Hardy paradox can be expressed in the following form: there is a local assignment (in the Hardy case, the assignment $a_1 \mapsto 0, b_1 \mapsto 0$) which is in the support, but which cannot be extended to a global assignment which is compatible with the support. This says that the support cannot be covered by the projections of global assignments. 
A stronger form of contextuality is when \emph{no global assignments are consistent with the support at all}.
Note that this stronger form does not hold for the Hardy paradox.

Several much-studied constructions from the quantum information literature exemplify strong contextuality.
An important example is the Popescu--Rohrlich (PR) box \cite{popescu:94} shown in Fig.~\ref{fig:HardyPR}.
%
%\begin{figure}
%\begin{center}
%\small
%
%
%\begin{tabular}{ll|ccccc}
%A & B & $(0, 0)$ & $(1, 0)$ & $(0, 1)$ & $(1, 1)$  &  \\ \hline
%$a_1$ & $b_1$ & $1$ & $0$ & $0$ & $1$ & \\
%$a_1$ & $b_2$ & $1$ & $0$ & $0$ & $1$ & \\
%$a_2$ & $b_1$ & $1$ & $0$ & $0$ & $1$ & \\
%$a_2$ & $b_2$ & $0$ & $1$ & $1$ & $0$ & 
%\end{tabular}
%\end{center}
%\caption{The PR Box}
%\label{fig:PR}
%\end{figure}%
This is a behaviour which satisfies the \emph{no-signalling principle} \cite{popescu:94}, meaning that the probability of Alice observing a particular outcome for her choice of measurement (e.g. $a_1=0$), is independent of whether Bob chooses measurement $b_1$ or $b_2$; and vice versa. 
That is, Alice and Bob cannot \emph{signal} to one another, enforcing compatibility with relativistic constraints.
%However, despite satisfying the no-signalling principle, the PR box does not admit a quantum realisation. Note that the full support of this model correspond to the propositions used in showing the contextuality of the Bell table from Fig.~\ref{fig:Bell}, and hence the fact that these propositions are not simultaneously satisfiable shows the strong contextuality of the model.
%
In fact, there is provably no bipartite quantum-realisable behaviour of this kind which is strongly contextual \cite{lal:11,mansfield:14}. 
However, as soon as we go to three or more parties, strong contextuality does arise from entangled quantum states, as we shall see in \S\ref{sec:AvN}.

\subsubsection*{Visualizing contextuality}

The tables which have appeared in our examples can be displayed in a visually appealing way which makes the fibred topological structure apparent, and forms an intuitive bridge to the formal development of the sheaf-theoretic ideas in the next section.

First, we look at the Hardy table from Fig.~\ref{fig:HardyPR}, displayed as a ``bundle diagram'' on the left of Fig.~\ref{fig:bundle}.
Note that all unspecified entries of the Hardy table are set to $1$.

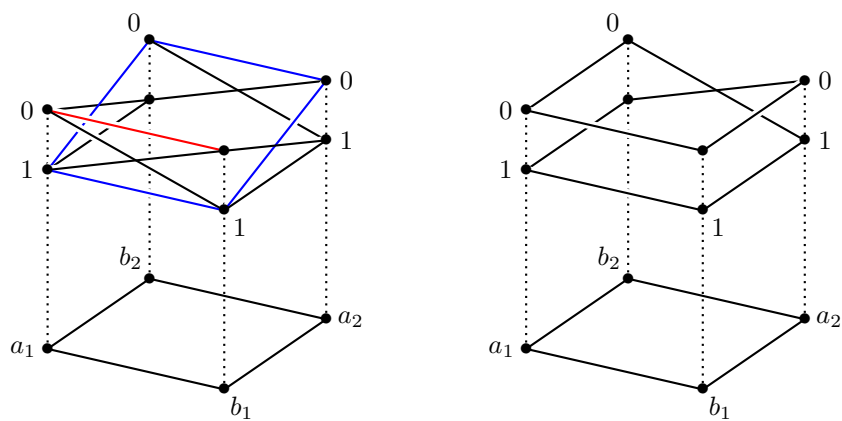
\begin{figure}
\begin{center}
\begin{tikzpicture}[x=45pt,y=45pt,thick,label distance=-0.25em,baseline=(O.base)]
\coordinate (O) at (0,0);
\coordinate (T) at (0,1.5);
\coordinate (u) at (0,0.5);
\coordinate [inner sep=0em] (v0) at ($ ({-cos(1*pi/12 r)*1.2},{-sin(1*pi/12 r)*0.48}) $);
\coordinate [inner sep=0em] (v1) at ($ ({-cos(7*pi/12 r)*1.2},{-sin(7*pi/12 r)*0.48}) $);
\coordinate [inner sep=0em] (v2) at ($ ({-cos(13*pi/12 r)*1.2},{-sin(13*pi/12 r)*0.48}) $);
\coordinate [inner sep=0em] (v3) at ($ ({-cos(19*pi/12 r)*1.2},{-sin(19*pi/12 r)*0.48}) $);
\coordinate [inner sep=0em] (v0-1) at ($ (v0) + (T) $);
\coordinate [inner sep=0em] (v0-0) at ($ (v0-1) + (u) $);
\coordinate [inner sep=0em] (v1-1) at ($ (v1) + (T) $);
\coordinate [inner sep=0em] (v1-0) at ($ (v1-1) + (u) $);
\coordinate [inner sep=0em] (v2-1) at ($ (v2) + (T) $);
\coordinate [inner sep=0em] (v2-0) at ($ (v2-1) + (u) $);
\coordinate [inner sep=0em] (v3-1) at ($ (v3) + (T) $);
\coordinate [inner sep=0em] (v3-0) at ($ (v3-1) + (u) $);
\draw (v0) -- (v1) -- (v2) -- (v3) -- (v0);
\draw [dotted] (v0-0) -- (v0);
\draw [dotted] (v1-0) -- (v1);
\draw [dotted] (v2-0) -- (v2);
\draw [dotted] (v3-0) -- (v3);
\node [inner sep=0.1em] (v0') at (v0) {$\bullet$};
\node [anchor=east,inner sep=0em] at (v0'.west) {$a_1$};
\node [inner sep=0.1em,label={[label distance=-0.625em]330:{$b_1$}}] at (v1) {$\bullet$};
\node [inner sep=0.1em] (v2') at (v2) {$\bullet$};
\node [anchor=west,inner sep=0em] at (v2'.east) {$a_2$};
\node [inner sep=0.1em,label={[label distance=-0.5em]175:{$b_2$}}] at (v3) {$\bullet$};
\draw [line width=3.2pt,white] (v0-0) -- (v3-1);
\draw [line width=3.2pt,white] (v0-1) -- (v3-0);
\draw [line width=3.2pt,white] (v0-1) -- (v3-1);
\draw (v0-0) -- (v3-1);
\draw [blue] (v0-1) -- (v3-0);
\draw (v0-1) -- (v3-1);
\draw [line width=3.2pt,white] (v2-0) -- (v3-0);
\draw [line width=3.2pt,white] (v2-0) -- (v3-1);
\draw [line width=3.2pt,white] (v2-1) -- (v3-0);
\draw [blue] (v2-0) -- (v3-0);
\draw (v2-0) -- (v3-1);
\draw (v2-1) -- (v3-0);
\draw [line width=3.2pt,white] (v0-0) -- (v1-0);
\draw [line width=3.2pt,white] (v0-0) -- (v1-1);
\draw [line width=3.2pt,white] (v0-1) -- (v1-0);
\draw [line width=3.2pt,white] (v0-1) -- (v1-1);
\draw [red] (v0-0) -- (v1-0);
\draw (v0-0) -- (v1-1);
\draw (v0-1) -- (v1-0);
\draw [blue] (v0-1) -- (v1-1);
\draw [line width=3.2pt,white] (v2-0) -- (v1-1);
\draw [line width=3.2pt,white] (v2-1) -- (v1-0);
\draw [line width=3.2pt,white] (v2-1) -- (v1-1);
\draw [blue] (v2-0) -- (v1-1);
\draw (v2-1) -- (v1-0);
\draw (v2-1) -- (v1-1);
\node [inner sep=0.1em,label=left:{$0$}] at (v0-0) {$\bullet$};
\node [inner sep=0.1em,label=left:{$1$}] at (v0-1) {$\bullet$};
\node [inner sep=0.1em] at (v1-0) {$\bullet$};
\node [inner sep=0.1em,label={[label distance=-0.5em]330:{$1$}}] at (v1-1) {$\bullet$};
\node [inner sep=0.1em,label=right:{$0$}] at (v2-0) {$\bullet$};
\node [inner sep=0.1em,label=right:{$1$}] at (v2-1) {$\bullet$};
\node [inner sep=0.1em,label={[label distance=-0.5em]150:{$0$}}] at (v3-0) {$\bullet$};
\node [inner sep=0.1em] at (v3-1) {$\bullet$};
\end{tikzpicture}
\hfil%
\begin{tikzpicture}[x=45pt,y=45pt,thick,label distance=-0.25em,baseline=(O.base)]
\coordinate (O) at (0,0);
\coordinate (T) at (0,1.5);
\coordinate (u) at (0,0.5);
\coordinate [inner sep=0em] (v0) at ($ ({-cos(1*pi/12 r)*1.2},{-sin(1*pi/12 r)*0.48}) $);
\coordinate [inner sep=0em] (v1) at ($ ({-cos(7*pi/12 r)*1.2},{-sin(7*pi/12 r)*0.48}) $);
\coordinate [inner sep=0em] (v2) at ($ ({-cos(13*pi/12 r)*1.2},{-sin(13*pi/12 r)*0.48}) $);
\coordinate [inner sep=0em] (v3) at ($ ({-cos(19*pi/12 r)*1.2},{-sin(19*pi/12 r)*0.48}) $);
\coordinate [inner sep=0em] (v0-1) at ($ (v0) + (T) $);
\coordinate [inner sep=0em] (v0-0) at ($ (v0-1) + (u) $);
\coordinate [inner sep=0em] (v1-1) at ($ (v1) + (T) $);
\coordinate [inner sep=0em] (v1-0) at ($ (v1-1) + (u) $);
\coordinate [inner sep=0em] (v2-1) at ($ (v2) + (T) $);
\coordinate [inner sep=0em] (v2-0) at ($ (v2-1) + (u) $);
\coordinate [inner sep=0em] (v3-1) at ($ (v3) + (T) $);
\coordinate [inner sep=0em] (v3-0) at ($ (v3-1) + (u) $);
\draw (v0) -- (v1) -- (v2) -- (v3) -- (v0);
\draw [dotted] (v0-0) -- (v0);
\draw [dotted] (v1-0) -- (v1);
\draw [dotted] (v2-0) -- (v2);
\draw [dotted] (v3-0) -- (v3);
\node [inner sep=0.1em] (v0') at (v0) {$\bullet$};
\node [anchor=east,inner sep=0em] at (v0'.west) {$a_1$};
\node [inner sep=0.1em,label={[label distance=-0.625em]330:{$b_1$}}] at (v1) {$\bullet$};
\node [inner sep=0.1em] (v2') at (v2) {$\bullet$};
\node [anchor=west,inner sep=0em] at (v2'.east) {$a_2$};
\node [inner sep=0.1em,label={[label distance=-0.5em]175:{$b_2$}}] at (v3) {$\bullet$};
\draw [line width=3.2pt,white] (v0-0) -- (v3-0);
\draw [line width=3.2pt,white] (v0-1) -- (v3-1);
\draw (v0-0) -- (v3-0);
\draw (v0-1) -- (v3-1);
\draw [line width=3.2pt,white] (v2-0) -- (v3-1);
\draw [line width=3.2pt,white] (v2-1) -- (v3-0);
\draw (v2-0) -- (v3-1);
\draw (v2-1) -- (v3-0);
\draw [line width=3.2pt,white] (v0-0) -- (v1-0);
\draw [line width=3.2pt,white] (v0-1) -- (v1-1);
\draw (v0-0) -- (v1-0);
\draw (v0-1) -- (v1-1);
\draw [line width=3.2pt,white] (v2-0) -- (v1-0);
\draw [line width=3.2pt,white] (v2-1) -- (v1-1);
\draw (v2-0) -- (v1-0);
\draw (v2-1) -- (v1-1);
\node [inner sep=0.1em,label=left:{$0$}] at (v0-0) {$\bullet$};
\node [inner sep=0.1em,label=left:{$1$}] at (v0-1) {$\bullet$};
\node [inner sep=0.1em] at (v1-0) {$\bullet$};
\node [inner sep=0.1em,label={[label distance=-0.5em]330:{$1$}}] at (v1-1) {$\bullet$};
\node [inner sep=0.1em,label=right:{$0$}] at (v2-0) {$\bullet$};
\node [inner sep=0.1em,label=right:{$1$}] at (v2-1) {$\bullet$};
\node [inner sep=0.1em,label={[label distance=-0.5em]150:{$0$}}] at (v3-0) {$\bullet$};
\node [inner sep=0.1em] at (v3-1) {$\bullet$};
\end{tikzpicture}
\end{center}
\caption{The Hardy table and the PR box as bundles.}
\label{fig:bundle}
\end{figure}

What we see in this representation is the \emph{base space} of the variables $a_1$, $a_2$, $b_1$, $b_2$. There is an edge between two variables when they can be measured together. The pairs of co-measurable variables correspond to the rows of the table. In terms of quantum theory, these correspond to pairs of \emph{compatible observables}. Above each vertex  is a \emph{fibre} of those values which can be assigned to the variable---in this example, $0$ and $1$ in each fibre. There is an edge between values in adjacent fibres precisely when the corresponding \emph{joint outcome} is possible, \ie has a $1$ entry in the table. Thus there are three edges for each of the pairs $\{ a_1, b_2 \}$, $\{ a_2, b_1 \}$ and $\{ a_2, b_2 \}$.

A \emph{global assignment} corresponds to a closed path traversing all the fibres exactly once. We call such a path \emph{univocal} since it assigns a unique value to each variable. Note that there is such a path, marked in blue; thus the Hardy model is not strongly contextual. However, there is no such path which includes the edge displayed in red. This shows the logical contextuality of the model.

Next, we consider the PR box displayed as a bundle on the right of Fig.~\ref{fig:bundle}.
In this case, the model is strongly contextual, and accordingly there is no univocal closed path.

\subsubsection*{Contextuality, logic and paradoxes}

The arguments for quantum contextuality we have discussed may be said to skirt the borders of paradox, but they do not cross those borders. The information we can gather from observing the co-measurable variables is locally consistent, but it cannot in general be pieced together into a globally consistent assignment of values to all the variables simultaneously. Thus we must give up the idea that physically observable variables have objective, ``real''  values independent of the measurement context being considered. This is very disturbing for our understanding of the nature of physical reality, but there is no direct contradiction between logic and experience. 
We shall now show that a similar analysis can be applied to some of the fundamental logical paradoxes.

% \subsubsection*{Liar cycles}
A \emph{Liar cycle} of length $N$ is a sequence of statements of the following kind.
\[
S_1 : S_2 \text{ is true,} \quad
S_2 : S_3 \text{ is true,} \quad
\dots \quad, \quad
S_{N-1} : S_N \text{ is true,} \quad
S_N : S_1 \text{ is false.}
\]
For $N=1$, this is the classic Liar sentence $S :  S \text{ is false}$.
These sentences contain two features which go beyond standard logic: references to other sentences, and a truth predicate.
While it would be possible to make a more refined analysis directly modelling these features, we will not pursue this here, noting that it has been argued extensively and rather compellingly in much of the recent literature on the paradoxes that the essential content is preserved by replacing statements with these features by \emph{boolean equations} \cite{wen:01,cook:04,walicki:09}. For the Liar cycles, we introduce boolean variables $x_1, \ldots , x_n$, and consider the equations
$x_1 = x_2$, \ldots , $x_{n-1} = x_n$, $x_n = \neg x_1$.
%\[ x_1 = x_2, \quad \dots \quad , \quad x_{n-1} = x_n, \quad x_n = \neg x_1 \Mdot \]
The ``paradoxical'' nature of the original statements is now captured by the inconsistency of these equations.

Note that we can regard each of these equations as fibered over the set of variables which occur in it:
\[
\{ x_1, x_2 \} :  x_1  {} = {}  x_2, \quad
%\{ x_2, x_3 \} :  x_2  {} = {}  x_3 \quad
\dots \quad, \quad
\{ x_{n-1}, x_n \} :  x_{n-1}  {} = {}  x_n, \quad
\{ x_n , x_1 \} :  x_n  {} = {}  \neg x_1 .
\]
Any subset of  up to $n-1$ of these equations is consistent; while the whole set is inconsistent.

Up to rearrangement, the Liar cycle of length 4 corresponds exactly to the PR box. The usual reasoning to derive a contradiction from the Liar cycle corresponds precisely to the attempt to find a univocal path in the bundle diagram on the right of Fig.~\ref{fig:bundle}.
To relate the notations, we make the following correspondences between the variables of Fig.~\ref{fig:bundle} and those of the boolean equations:
%\[ x_1 \sim a_2, \;\; x_2 \sim b_1, \;\; x_3 \sim a_1, \;\; x_4 \sim b_2  \Mdot \]
$x_1 \sim a_2$, $x_2 \sim b_1$, $x_3 \sim a_1$, $x_4 \sim b_2$.
Thus we can read the equation $x_1 = x_2$ as ``$a_2$ is correlated with $b_1$'', and $x_4 = \neg x_1$ as ``$a_2$ is anti-correlated with $b_2$''.

Now suppose that we try to set $a_2$ to $1$. Following the path in Fig.~\ref{fig:bundle} on the right leads to the following local propagation of values:
\begin{gather*}
a_2 = 1 \rsqa b_1 = 1 \rsqa a_1 = 1 \rsqa b_2 = 1 \rsqa a_2 = 0 \\
a_2 = 0 \rsqa b_1 = 0 \rsqa a_1 = 0 \rsqa b_2 = 0 \rsqa a_2 = 1
\end{gather*}
The first half of the path corresponds to the usual derivation of a contradiction from the assumption that $S_1$ is true, and the second half to deriving a contradiction from the assumption that $S_1$ is false.

We have discussed a specific case here, but the analysis can be generalised to a large class of examples along the lines of \cite{cook:04,walicki:09}. The tools from sheaf cohomology which we will develop in the remainder of the paper can be applied to these examples. We plan to give an extended treatment of these ideas in future work.

\section{Sheaf formulation of contextuality}\label{sec:sheafframework}

In this section we summarise the main ideas of the sheaf-theoretic formalism from \cite{abramsky:11}.
In \S\ref{ssec:bundle} we saw that logical contextuality can be expressed in terms of a bundle of outcomes over a base space of measurements and contexts.
This idea can be formalized by regarding the bundle as a  sheaf.

For our purposes, it will be sufficient to view the base space as the discrete space on a finite set $X$ of variables.\footnote{The fact that our examples involve a discrete base space $X$ does not trivialise our approach, and certainly does not mean that we are taking the cohomology of a discrete space!  It is standard that in a topological bundle, the interesting twisting occurs in the fibres, not in the base. A classic example is the M\"obius strip, displayed as a fibre bundle over the circle. The circle is not twisted! In our case, it is clear from our examples that non-trivial twisting does occur. Moreover, our results in Section~\ref{sec:cohomAvN} will clearly show the non-triviality of our cohomological obstructions.}
In the quantum case, these variables will be labels for measurements.
The measurement contexts will be represented by a family $\M = \{ C_i \}_{i \in I}$ of subsets of $X$. These are the sets of variables which can be measured together---%
in quantum terms, the compatible families of observables.
We assume that $\M$ covers $X$, \ie $\bigcup \M = X$;
hence we call $\M$ a \emph{measurement cover}. We shall also assume that $\M$ forms an antichain, so these are the \emph{maximal contexts}.
We  also assume that all the variables have the same fibre, $O$, of values or outcomes that can be assigned to them.
Such a triple $\tuple{X, \M, O}$ is called a \emph{measurement scenario}.
We define a presheaf of sets over $\mathcal{P}(X)$, namely $\fdef{\EE}{U}{O^U}$ with restriction
$\fdecdef{\EE(U \subseteq U')}{\EE(U')}{\EE(U)}{s}{s |_U}$.
This presheaf $\EE$ is in fact a sheaf, called the \emph{sheaf of events}.
Each $s \in \EE(U)$ is a \emph{section}, and, in particular, $g \in \EE(X)$ is a \emph{global section}. 

Note that a probability table can be represented by a family $\{ p_C \}_{C \in \M}$ with $p_C$ a probability distribution on $\EE(C) = O^{C}$,
where contexts $C$ correspond to the rows of the table. 
Similarly,  ``possibility tables'' such as the Hardy model and the PR box (Figs.~\ref{fig:HardyPR} and \ref{fig:bundle}) can be represented by boolean distributions. 
This latter case,
with which the logical and strong forms of contextuality are concerned,
can equivalently be represented by a subpresheaf $\Se$ of $\EE$, where
for each context $U \subseteq X$, $\Se(U) \subseteq O^U$ is the set of all possible outcomes.
Explicitly, $\Se$ is defined as follows,
where $\supp(p_C |_{U \cap C})$ is the support of the marginal of $p_C$ at $U \cap C$.
\[ \Se(U)  \defeq  \setdef{s \in O^U}{\Forall{C \in \M} s |_{U \cap C} \in \supp(p_C |_{U \cap C})} \]

\noindent Abstracting from this situation,
we assume we are dealing with a sub-presheaf $\Se$ of $\EE$ with the following properties:
\begin{enumerate}
\def\theenumi{E\arabic{enumi}}
\item\label{def:empirical.model.1}
$\Se(C) \neq \varnothing$ for all $C \in \M$
\end{enumerate}
(\ie that any possible joint measurement yields some joint outcome), and moreover that
\begin{enumerate}
\def\theenumi{E\arabic{enumi}}
\addtocounter{enumi}{1}
\item\label{def:empirical.model.2}
  $\Se$ is \emph{flasque beneath the cover}, meaning that $\fdec{\Se(U \subseteq U')}{\Se(U')}{\Se(U)}$ is surjective whenever
$U \subseteq U' \subseteq C$ for some $C \in \M$,
\end{enumerate}
which by \cite{abramsky:11} amounts to saying that the underlying empirical model satisfies no-signalling.
\begin{enumerate}
\def\theenumi{E\arabic{enumi}}
\addtocounter{enumi}{2}
\item\label{def:empirical.model.3}
A \emph{compatible family} for the cover $\M$ is a family $\{ s_C \}_{C \in \M}$ with  $s_C \in \Se(C)$, and such that,
for all $C, C' \in \M$: $s_C |_{C \cap C'} = s_{C'} |_{C \cap C'}$.
%\[\Forall{C, C' \in \M} s_C |_{C \cap C'} = s_{C'} |_{C \cap C'} \Mdot\]
We assume that such a family induces a global section in $\Se(X)$. (This global section must be unique, since $\Se$ is a subpresheaf of $\EE$, hence separated).
\end{enumerate}

\noindent
What these conditions say is that $\Se$ is determined by its values $\Se(C)$ at the contexts $C \in \M$, below $\M$ by being flasque, and above $\M$ by the sheaf condition.

\begin{definition}
By an \emph{empirical model} on $\tuple{X, \M, O}$, we mean a subpresheaf $\Se$ of $\EE$ satisfying \eqref{def:empirical.model.1}, \eqref{def:empirical.model.2}, and \eqref{def:empirical.model.3}.
\end{definition}

In \cite{abramsky:11}, we used the term ``empirical model'' for the probability table $\{ p_C \}_{C \in \M}$. In the present paper, we shall only work with the associated support presheaf $\Se$, and so it is more convenient to refer to this as the model.

We can use this formalisation to characterize contextuality as follows.

\begin{definition}
For any empirical model $\Se$:
\begin{itemize}
  \item
For $C \in \M$ and $s \in \Se(C)$, $\Se$ is logically contextual at $s$, written $\LC(\Se,s)$, if $s$ belongs to no compatible family.
$\Se$ is \emph{logically contextual}, written $\LC(\Se)$, if $\LC(\Se,s)$ for some $s$.
  \item $\Se$ is \emph{strongly contextual}, written $\SC(\Se)$, if $\LC(\Se,s)$ for all $s$. Equivalently, it is strongly contextual if it has no global section, \ie if $\Se(X) = \varnothing$.
\end{itemize}
\end{definition}

Note that, for every probability table $\{ p_C \}_{C \in \M}$ that satisfies the no-signalling principle,  the supports of the distributions $p_C$ induce an empirical model $\Se$, and therefore  logical or strong contextuality can be characterized as above.
This formulation of contextuality makes it natural to use sheaf cohomology, as we will see in \S\ref{sec:cohom}.

\section{All-vs-Nothing arguments}\label{sec:AvN}
Quantum theory provides many instances of strong contextuality.
Among the first to observe quantum strong contextuality (though not in the general terms that we do) was Mermin \cite{mermin:90},
who showed the GHZ state to be strongly contextual
using a kind of argument he dubbed `all versus nothing'.
We show in \S\ref{ssec:AvN-quantum}
that arguments of this type can in fact be used to prove strong contextuality
for a large class of states in quantum theory,
particularly in stabiliser quantum mechanics, which plays a crucial r\^ole in quantum computation.
Moreover, in \S\ref{ssec:AvN-defs},
we give a much more general formulation of this type of argument that can be used to show strong contextuality for a much larger class of models.

\subsection{All-vs-Nothing for quantum theory}\label{ssec:AvN-quantum}

The GHZ state is a tripartite state of qubits, defined as $(\ket{\ua\ua\ua}   + \ \ket{\da\da\da})/\SQT$.
We assume that each party $i = 1, 2, 3$ can perform Pauli measurements in $\{X_i, Y_i\}$,
and each measurement has outcomes in $O = \ZZ_2 = \{0,1\}$.%
\footnote{
Although the eigenvalues of the Pauli matrices are $+1$ and $-1$, we relabel $+1$, $-1$, $\times$ as $0$, $1$, $\oplus$, respectively. The eigenvalues of a joint measurements $A_1 \otimes A_2 \otimes A_3$ are the products of the eigenvalues of the measurements at each site, so they are also $\pm1$. As such, in the usual representation, these joint measurements are still dichotomic and only distinguish joint outcomes up to parity. Mermin's argument shows that this information is sufficient to derive strong contextuality.
}
Then, following Mermin's argument, the possible joint outcomes satisfy these parity equations:
%
% \begin{align*}
% \begin{array}{c@{}c@{}c@{}c@{}c@{\,}c@{\,}c}
% X_1 & {} \oplus {} & Y_2 & {} \oplus {} & Y_3 & {} = {} & 0 \\
% Y_1 & {} \oplus {} & Y_2 & {} \oplus {} & X_3 & {} = {} & 0 \\
% \end{array}
% &&
% \begin{array}{c@{}c@{}c@{}c@{}c@{\,}c@{\,}c}
% Y_1 & {} \oplus {} & X_2 & {} \oplus {} & Y_3 & {} = {} & 0 \\
% X_1 & {} \oplus {} & X_2 & {} \oplus {} & X_3 & {} = {} & 1
% \end{array}
% \end{align*}
%
\begin{align*}
& X_1 \oplus Y_2 \oplus Y_3 = 1 , &
& Y_1 \oplus Y_2 \oplus X_3 = 1 , &
& Y_1 \oplus X_2 \oplus Y_3 = 1 , &
& X_1 \oplus X_2 \oplus X_3 = 0 .
\end{align*}
These equations are inconsistent because,
regardless of the outcomes assigned to the observables $X_1, \dots, Y_3$,
the left-hand sides sum to $0$ (since each variable occurs twice) whereas the right-hand sides sum to $1$.
This shows that the model is strongly contextual,
as there is no global assignment of outcomes to observables consistent with the observed local assignments.

The essence of the argument is that the possible local assignments satisfy  systems of parity equations that
admit no global solution. We call this an \emph{All-vs-Nothing argument}.
In fact, such
arguments arise naturally from a much larger class of states in stabiliser quantum theory \cite{nielsen:00}.
Consider the Pauli $n$-group $\mathcal{P}_n$,
whose elements are $n$-tuples of Pauli operators (from $\{ X, Y, Z, I \})$ with a global phase from $\{ \pm 1, \pm i \}$.

\begin{definition}
 An \emph{AvN triple} in $\mathcal{P}_n$ is a triple $\langle e, f, g \rangle$ of elements of $\mathcal{P}_n$ with global phases $+1$, which pairwise commute, and which satisfy the following conditions:
\begin{enumerate}
	\def\theenumi{A\arabic{enumi}}
	\item\label{def:AvN.triple.1}
	For each $i=1, \ldots ,n$, at least two of $e_i$, $f_i$, $g_i$ are equal.
	\item\label{def:AvN.triple.2}
	The number of $i$ such that $e_i = g_i \neq f_i$, all distinct from $I$, is odd. 
\end{enumerate}
\end{definition}

Mermin's argument, and the other All-vs-Nothing arguments which have appeared in the literature, can be seen to come down to exhibiting AvN triples.

\begin{theorem}\label{thm:pauavn}
  Let $S$ be the subgroup of $\mathcal{P}_n$ generated by an AvN triple, and $V_S$ the subspace stabilised by $S$.
For every state $\ket{\psi}$ in $V_S$, the empirical model realised by $\ket{\psi}$ under the Pauli measurements
admits an All-vs-Nothing argument.
\end{theorem}
\begin{proof} %[Proof of Theorem \ref{thm:pauavn}]
First, we recall the quantum mechanics formula for the expected value of an observable $A$ on a state $v$:
\[ \langle A \rangle_v \; = \; \langle v | A | v \rangle . \]
Note that
\[ \langle v | A | v \rangle = 1 \IFF A \ket{v} = \ket{v} . \]
Thus $A$ stabilises the state $v$ iff the expected value is $1$. Suppose that $A$ is a dichotomic observable, with eigenvalues $+1$, $-1$ (see footnote~3), and $v$ is a state it stabilises. The support of the distribution on joint outcomes obtained by measuring  $A$ on $v$ must contain only outcomes of  even parity; while if $-A$ stabilises $v$, then the support will contain only outcomes of odd parity.
If $A = P_1 \cdots P_n$ in $\mathcal{P}_n$, the former case translates into the equation
\[ x_1 \oplus \cdots \oplus x_n = 0 \]
where we associate the variable $x_i$ with $P_i$; while in the latter case, it corresponds to the equation
\[ x_1 \oplus \cdots \oplus x_n = 1 . \]
If the set of equations satisfied by a state $v$ stabilised by a subgroup $S$ of $\mathcal{P}_n$ is inconsistent, we say that $v$ \emph{admits an  All-vs-Nothing argument} with respect to the measurements $h$ with global phase $+1$ such that either $h$ or $-h$ is in $S$.

We now show that any state $v$ in the subspace $V_S$ stabilised by the subgroup $S$ generated by an AvN triple $\langle e, f, g \rangle$ admits an All-vs-Nothing argument.
First, by the algebra of the Pauli matrices, we see from \eqref{def:AvN.triple.1} that if $\{ e_i, f_i, g_i \} = \{ P, Q \}$, with at least two equal to $P$, the componentwise product $e_i f_i g_i$ will, disregarding global phase, be $Q$. By \eqref{def:AvN.triple.2}, we see that the product $efg = -h$, an element of $\mathcal{P}_n$ with global phase $-1$, which translates into a condition of odd parity on the support of any state stabilised by these operators for the measurement $h$.
On the other hand, condition \eqref{def:AvN.triple.1} implies that under any global assignment to the variables, we can cancel the repeated items in each column, and deduce an even parity for $h$.
\end{proof}

If $e,f,g$ have linearly independent check vectors, they generate a subgroup $S$ such that $V_S$ has dimension $2^{n-3}$ \cite{nielsen:00,Caves06}. Thus we obtain a large class of states admitting All-vs-Nothing arguments.

\subsection{Generalized All-vs-Nothing arguments}\label{ssec:AvN-defs}

Despite their established use in the quantum literature,
All-vs-Nothing arguments as considered above do not exist for all strongly contextual models.
However, a natural generalization applies to more models.
%% ; see e.g.~Example~\ref{avnexample} in Appendix A.
%
For instance, the model called `box 25' in \cite{pironio:11} admits no parity AvN argument,
but it still satisfies the following equations, in which the coefficients of each variable on the left-hand sides add up to a multiple of $3$, whereas the right-hand sides do not:
\begin{align*}
a_0 + 2 b_0 & \equiv 0 \text{ mod } 3 &
a_1 + 2 c_0 & \equiv 0 \text{ mod } 3 \\
a_0 + b_1 + c_0 & \equiv 2 \text{ mod } 3 &
a_0 + b_1 + c_1 & \equiv 2 \text{ mod } 3 \\
a_1 + b_0 + c_1 & \equiv 2 \text{ mod } 3 &
a_1 + b_1 + c_1 & \equiv 2 \text{ mod } 3
\end{align*}
%
%%This suggests using a general ring instead of just $\ZZ_2$.
This example suggests using a general $\ZZ_n$ instead of just $\ZZ_2$.
But once we realize that it is the ring structure of $\ZZ_n$ which plays the key r\^ole,
we can obtain an even more general version.

Fix a ring\footnote{All rings considered in this paper will be commutative and with unit.} $R$, and a measurement scenario $\tuple{X, \M, R}$.  
\begin{definition}\label{def:lineqs}
An \emph{$R$-linear equation} is a triple $\phi = \tuple{C, a, b}$ with $C \in \M$, $\fdec{a}{C}{R}$ and $b \in R$.
Write $V_{\phi} := C$. 
An assignment $s \in \EE(C)$ \emph{satisfies} $\phi$, written $s \models \phi$, if
  \[
  \sum_{m \in C} a(m) s(m) = b
  \Mdot \]
This lifts to the level of systems of equations, or theories, and sets of assignments, or ``models'':
\begin{itemize}
\item A system of equations $\Gamma$ has a set of satisfying assignments,
  $\Sol(\Gamma) \defeq \setdef{s \in \EE(C)}{\Forall{\phi \in \Gamma} s \models \phi}$.
\item A set of assignments $S \subseteq  \EE(C)$ determines an $R$-linear theory,
  $\Th_R(S) \defeq \setdef{\phi}{\Forall{s \in S} s \models \phi}$.
\end{itemize}
\end{definition}

\begin{definition}\label{def:AvN.R}
 Given an empirical model $\Se$, define its \emph{$R$-linear theory} to be
 \[\Th_R(\Se) \defeq \bigcup_{C \in \M} \Th_R(\Se(C)) = \setdef{\phi}{\Forall{s \in \Se(V_\phi)} s \models \phi}\Mdot\]
We say that $\Se$ is \emph{$\AvN_R$}, written $\AvN_R(\Se)$, if $\Th_R(\Se)$ is inconsistent, meaning that there is no global assignment $\fdec{g}{X}{R}$ such that $\Forall{\phi \in \Th_R(\Se)} g |_{V_\phi} \models \phi$.
\end{definition}

\begin{proposition}\label{prop:AvNimpliesSC}
An $\AvN_R$ model is strongly contextual.
\end{proposition}
\begin{proof}
%{\color{orange}
Suppose $\Se$ is not strongly contextual, \ie that there is some $g \in \Se(X)$.
Then, for each $\phi \in \Th_R(\Se)$, $g |_{V_\phi} \in \Se(V_\phi)$, hence $g |_{V_\phi} \models \phi$.
Thus, $\Th_R(\Se)$ is consistent.
%}
\end{proof}

\subsection{Affine closures}\label{ssec:AvN-aff}
We now consider the relationship between $R$-linear theories and empirical models more closely.
%We shall relate the existence of an $\AvN_R$ argument
%to the 
%
%strong contextuality of a derived model.

First, we focus on a single context, or set of variables, $U \subseteq X$. 
The maps between theories and models, 
%$\Th : \ps \EE(U) \longrightarrow \cat{Theories}$ and $\Sol : \cat{Theories} \longrightarrow \ps \EE(U)$,
% \[ \Th : \ps \EE(U) \longrightarrow \cat{Theories}, \qquad \Sol : \cat{Theories} \longrightarrow \ps \EE(U) \]
%\[
%\xymatrix@!C=.8cm{
%\ps \EE(U) \ar@/^1pc/[rr]^{\Th} & ~ & \cat{Theories} \ar@/^1pc/^{\Sol}[ll]
%}
%\Mcomma
%\]
$\Th : \ps \EE(U) \mathrel{\substack{\xrightarrow{\rule{12pt}{0pt}} \\[-0.8ex] \xleftarrow{\rule{12pt}{0pt}}}} \cat{Theories} : \Sol$,
form a Galois connection,
$S \subseteq \Sol(\Gamma)$ iff $\Th(S) \supseteq \Gamma$,
corresponding to the lifting of the satisfaction relation to the powersets.
% \[
% S \subseteq \Sol(\Gamma) \quad\text{ iff }\quad \Th(S) \supseteq \Gamma \Mdot
% \]

%
%

We consider the closure operator $\Sol \circ \Th$, which gives the largest set of assignments whose theory is still the same.
First, note that $\EE(U) = R^U$ is a (free) $R$-module (and, when $R$ is a field, a vector space over $R$).
Given solutions $s_1, \ldots, s_t$ to a
linear equation, an affine combination of them
%  \[
% c_1 s_1 + \cdots + c_t s_t 
% \quad\text{ such that }\quad c_1 + \cdots + c_t = 1 \Mcomma
% \]
 is again a solution\footnote{{Affineness} is required because the equations may be inhomogeneous.}.
In other words, the set of solutions $\Sol(\Gamma)$ to a system of equations $\Gamma$ %
is an affine submodule of $\EE(U)$.
This means that $\affN \leq \Sol \circ \Th$,
%\begin{equation}\label{eq:aff<TM}
%  \affN \;\leq\; \Sol \circ \Th \Mcomma
%\end{equation}
where $\aff S$ stands for the \emph{affine closure} of a set $S \subseteq \EE(U)$:
 \[
 \aff S\, \defeq\, \setdef{\sum_{i=1}^{t} c_i s_i}{s_i \in S, c_i \in R, \sum_{i=1}^{t} c_i =1} \Mdot
 \]
In the particular case of vector spaces (\ie when $R$ is a field), then $\affN = \Sol \circ \Th$;
%\eqref{eq:aff<TM} is in fact an equality:
affine subspaces are exactly the possible solution sets of a theory,
and there cannot exist two different affine subspaces with the same theory,
as may happen for general rings $R$. 

We now lift this discussion to the level of empirical models.
The natural thing to do is to take the affine closure at each context $C\in \M$. However,
one must be careful to ensure that this yields a well-defined empirical model.
First, note that the affine closure operation above is natural on $U$: it gives a natural transformation
$\fdec{\aff}{\ps \circ \EE}{\ps \circ \EE}$, meaning that
\begin{equation}\label{eq:affnatural}(\aff S)|_{U'} = \aff(S|_{U'}) \Mdot\end{equation}
This follows easily from the coordinatewise definitions of the module operations on $\EE(U)$.

\begin{definition}
  Let $\Se$ be an empirical model on the scenario $\tuple{X,\M,R}$.
  We define its \emph{affine closure}, $\Aff \Se$,
  as the empirical model
%   given, at each $C \in \M$,
% %
% by
%   \[
%   \SAffe(C) \,\defeq\, \aff (\Se(C)) \Mdot 
%   \]
  given by $\SAffe(C) \,\defeq\, \aff (\Se(C))$ at each $C \in \M$.
 \end{definition}
The property \eqref{eq:affnatural} guarantees that $\Aff\Se$
can be consistently defined to be flasque below the cover as 
$\SAffe(U) = \aff (\Se(U))$.
This equality, however, does not hold for $U$ above the cover.
In particular, it may be that $\Se(X) = \vn$ ($\Se$ strongly contextual), but $\SAffe(X) \neq \vn$. 

Since $\Th_R(\Se)$ is given as the union of the theories at each maximal context,
the Galois connection above lifts to the level of empirical models.
One can therefore relate the notion of $\Se$ being $\AvN_R$ to the strong contextuality of the affine closure of $\Se$.

\begin{proposition}\label{prop:AvN->SCAff}
Let $\Se$ be an empirical model on $\tuple{X,\M,R}$. Then,
%  \[ \AvN(\Se) \;\implies\; \SC(\Aff \Se) \Mdot\]
$\AvN(\Se) \implies \SC(\Aff \Se)$.
If $R$ is a field, the converse also holds.
\end{proposition}
\begin{proof} %%[Proof of Proposition \ref{prop:AvN->SCAff}]
%{\color{orange}
From $\affN \leq \Sol \circ \Th$, $\Th_{R} (\Se) = \Th_{R} (\Aff \Se)$.
Hence, if $\Se$ is $\AvN_R$, then so is $\Aff \Se$,
implying by Proposition~\ref{prop:AvNimpliesSC} that it is strongly contextual.

For the converse in the case that $R$ is a field, suppose that $\Th_{R} (\Se)$ is consistent.
This means that there is a global assignment $\fdec{g}{X}{R}$ satisfying all the equations in $\Th_{R} (\Se)$.
But since for fields
$\Sol \Th_{R} (\Se) = \Aff \Se$,
we have that $g \in \SAffe(X)$,
hence the model $\Aff \Se$ is not strongly contextual.
%
%
%}
\end{proof}

\section{Cohomology witnesses contextuality}\label{sec:cohom}
The logical forms of contextuality are characterised by the existence of
obstructions to the extension of local sections to global compatible families.
Thus, it seems natural to apply the tools of sheaf cohomology,
which are well-suited to identifying obstructions of this kind,
in order to provide cohomological witnesses for contextuality.
This idea was put forward in previous work by the authors \cite{abramsky:11a},
the main points of which we now summarise.
In the next section, we shall prove that such cohomological witnesses of contextuality exist for the whole class of $\AvN_R$ models.

\subsection{\Cech\ cohomology}\label{ssec:cohom-cech}
Let $X$ be
a topological space, $\M$ be an open cover of $X$, and let
$\fdec{\FF}{\mathcal{O}(X)^\op}{\AbGrp}$ be a presheaf of abelian groups on $X$.
We shall be particularly concerned with the case where
$X$ is a set of measurements,
and $\M$ is the cover of maximal contexts of a measurement scenario.

\begin{definition}
  A \emph{$q$-simplex of the nerve} of $\M$,
  is a tuple $\sigma = \tuple{C_0, \ldots, C_q}$
  of elements of $\M$ with non-empty intersection, $|\sigma| \defeq \cap_{i=0}^q C_i \neq \vn$.
  We write $\Nerve(\M)^q$ for the set of $q$-simplices.
\end{definition} 
The $0$-simplices are simply the elements of the cover $\M$,
and the $1$-simplices are the pairs $\tuple{C_i,C_j}$ of intersecting elements of the cover. 
Given a $q+1$-simplex $\sigma = \tuple{C_0, \ldots , C_{q+1}}$, we can obtain $q$-simplices
\[ \bdmap_j (\sigma) \defeq  \tuple{C_0, \ldots, \widehat{C_j}, \ldots , C_{q+1}}, \qquad 0 \leq j \leq q+1 \]
by omitting one of the elements.
Note that
$| \sigma | \; \subseteq \; | \bdmap_j(\sigma) |$,
and so the presheaf $\FF$ has a restriction map 
$\fdec{\rmap{|\bdmap_j(\sigma)|}{|\sigma|}}{\FF(|\bdmap_j(\sigma)|)}{\FF(|\sigma|)}$.
% \[\fdec{\rmap{|\bdmap_j(\sigma)|}{|\sigma|}}{\FF(|\bdmap_j(\sigma)|)}{\FF(|\sigma|)} \Mdot\]

We now define the \emph{\Cech\ cochain complex}.
\begin{definition}
For each $q \geq 0$, the abelian group of \emph{$q$-cochains} is defined by:
\[ \Ccoch{q} \; \defeq \; \prod_{\sigma \in \Nerve(\M)^q} \FF( | \sigma |) \Mdot \]
The \emph{$q$-coboundary map},
$\fdec{\cobd^{q}}{\Ccoch{q}}{\Ccoch{q+1}}$,
is defined as follows:
for each $\omega = (\omega(\tau))_{\tau \in \Nerve(\M)^q} \in \Ccoch{q}$, and $\sigma \in \Nerve(\M)^{q+1}$,
\[ \cobd^{q}(\omega)(\sigma) \; \defeq \; \sum_{j = 0}^{q+1} (-1)^j \rmap{|\bdmap_j(\sigma)|}{|\sigma|}\omega(\bdmap_j \sigma) \Mdot \]
The \emph{augmented \Cech\ cochain complex} is the sequence
\[\xymatrix{\zero \ar[r] & \Ccoch{0} \ar[r]&  \Ccoch{1} \ar[r] &  \cdots} \Mdot \]
\end{definition}

\begin{proposition}
\label{boundprop}
For each $q$, $\cobd^q$ is a group homomorphism, and we have $\cobd^{q+1} \circ \cobd^q = 0$.
\end{proposition}

\begin{definition}
  For each $q\geq 0$, we define:
  \begin{itemize}
   \item  the \emph{$q$-cocycles} $\Cocyc{q} \defeq \ker\cobd^q$;
   \item  the \emph{$q$-coboundaries} $\Cobound{q} \defeq \img\cobd^{q-1}$.
 \end{itemize}
 \end{definition}
By Proposition~\ref{boundprop},
these are subgroups of $\Ccoch{q}$ with $\Cobound{q} \subseteq \Cocyc{q}$.
\begin{definition}
$\Cohom{q}$, the \emph{$q$-th \Cech\ cohomology group}, is defined as the quotient $\Cocyc{q} / \Cobound{q}$.
\end{definition}

Note that $\Cobound{0} = \zero$, and hence $\Cohom{0} \cong \Cocyc{0}$.
A $0$-cochain $\omega$ is a family $\enset{r_C \in \FF(C)}_{C \in \M}$.
Since, for each $1$-simplex $\sigma = (C, C')$,
\[ \cobd^{0}(\omega)(\sigma) \; = \; r_{C}|_{C \cap C'} \; - \; r_{C'}|_{C \cap C'} \Mcomma \]
$\omega$ is a cocycle (\ie satisfies $\cobd^{0}(c) = 0$) if and only if $r_{C}|_{C \cap C'} = r_{C'}|_{C \cap C'}$
%\[ r_{C}|_{C \cap C'} \; = \; r_{C'}|_{C \cap C'} \]
for all maximal contexts $C, C' \in \M$ with non-empty intersection.\footnote{%See the second Remark in Appendix A.
The condition for a $0$-cochain  $\omega = \{r_C\}$ to be a cocycle almost states that $r$ is a compatible family, except that it does not require compatibility
over restrictions to the empty context.
For our present purposes, we are only interested in connected covers (since one can always reduce the analysis of a scenario to its connected components), in which case the exception is irrelevant.
This is because, given any two contexts $C$ and $C'$ with empty intersection, there exists a sequence of contexts
\[C = C_0, \; C_1, \; \ldots, \; C_n = C' \]
such that $C_i \cap C_{i+1} \neq \vn$ for all $i$.
Then, for any $i$, we have 
\[r_{C_i}|_\vn = r_{C_i}|_{C_i \cap C_{i+1}}|_\vn = r_{C_{i+1}}|_{C_i \cap C_{i+1}}|_\vn = r_{C_{i+1}}|_\vn \Mdot\]
Consequently, 
$r_{C}|_{C \cap C'} = r_{C}|_\vn = r_{C'}|_\vn = r_{C'}|_{C \cap C'}$,
and so the family is compatible.
}

\subsection{Relative cohomology and obstructions}\label{ssce:cohom-relative}
In order to solve the problem of extending a local section to a global compatible family,
we need to consider the relative cohomology of $\FF$ with respect to an open subset $U \subseteq X$. We will assume that the presheaf is flasque beneath the cover (as is the case with $\Se$).

We define two auxiliary presheaves related to $\FF$.
First, $\FF |_U$ is defined by
\[ \FF |_U (V) \defeq \FF(U \cap V) \Mdot \]
There is an evident presheaf map $\fdec{p}{\FF}{\FF |_U}$ given as
\[  \fdecdef{p_V}{\FF(V)}{\FF(U \cap V)}{r}{r|_{U \cap V}} \Mdot \]
Secondly, $\FU$ is defined by $\FU(V) \defeq \ker(p_V)$. Thus, we have an exact sequence of presheaves%
\begin{equation}\label{eq:sequence}
\xymatrix{
 \zero \ar[r] & \FU \ar[r] & \FF \ar[r]^-{p} & \FF |_U 
} \Mdot
\end{equation}
The \emph{relative cohomology of $\FF$ with respect to $U$} is defined to be the cohomology of the presheaf $\FU$.

We now see how this can be used to identify \emph{cohomological obstructions} to the extension of a local section.
First, recall that in light of Proposition~\ref{boundprop}, the image of $\cobd^{0}$,
$\Cobound{1}$, is contained in $\Cocyc{1}$. Therefore, the map $\cobd^0$ can be corestricted to a map
$\fdec{\tilde{\cobd}^0}{\Ccoch{0}}{\Cocyc{1}}$, whose kernel is $\Cocyc{0} \cong \Cohom{0}$ and whose cokernel is $\Cocyc{1}/\Cobound{1} \cong \Cohom{1}$. In summary, we have:
\[
\xymatrix{
\Cohom{0}
\ar[r]^{\ker \tilde{\cobd}^0} &
\Ccoch{0}
\ar[r]^{\tilde{\cobd}^0} &
\Cocyc{1}
\ar[r]^{\coker \tilde{\cobd}^0} &
\Cohom{1}
}
\Mdot
\] 
We now lift the exact sequence of presheaves (\ref{eq:sequence}) considered above to the level of cochains. The map $\Ccoch{0} \longrightarrow \CcochRes{0}{U}$ is surjective due to flaccidity beneath the cover.
Putting this together with the previous observation, we obtain the diagram below:
\[
\xymatrix@C=2em{
\zero
\ar[r]
&
 \CcochRel{0}{U}
\ar[r]
\ar[d]_{\tilde{\cobd}^0}
&
 \Ccoch{0}
\ar[r]
\ar[d]_{\tilde{\cobd}^0}
&
 \CcochRes{0}{U}
\ar[r]
\ar[d]_{\tilde{\cobd}^0}
&
\zero
\\
\zero
\ar[r]
&
 \CocycRel{1}{U}
\ar[r]
&
 \Cocyc{1}
\ar[r]
&
 \CocycRes{1}{U}
&
}
\]
whose two rows are short exact sequences. The \emph{snake lemma} of homological algebra says that there exists a \emph{connecting homomorphism}
turning the kernels of the first row followed by the cokernels of the second into a long exact sequence, as shown by the following diagram.
\[
\xymatrix{
&&
\CohomRel{0}{U}
\ar[r]
\ar[d]
&
\Cohom{0}
\ar[r]
\ar[d]
&
\CohomRes{0}{U}
\ar[d]
\ar `r[rr]`[dd]`^d[llll]`[dddd] [ddddll]
&&
\\
&
\zero
\ar[r]
&
\CcochRel{0}{U}
\ar[r]
\ar[dd]
&
\Ccoch{0}
\ar[r]
\ar[dd]
&
\CcochRes{0}{U}
\ar[r]
\ar[dd]
&
\zero
&
\\
&&&&&&
\\
&
\zero
\ar[r]
&
\CocycRel{1}{U}
\ar[r]
\ar[d]
&
\Cocyc{1}
\ar[r]
\ar[d]
&
\CocycRes{1}{U}
% \ar[r]
\ar[d]
&
% \zero
&
\\
&&
\CohomRel{1}{U}
\ar[r]
&
\Cohom{1}
\ar[r]
&
\CohomRes{1}{U}
&&
}
\]

We are interested in the case where $U$ is an element $C_0$ of the cover $\M$.
Then $\CohomRes{0}{C_0}$ is clearly isomorphic to $\FF(C_0)$, meaning that its elements are the local sections at $C_0$.
\begin{definition}\label{def:obs}
Let $C_0$ be an element of the cover $\M$ and $r_0 \in \FF(C_0)$.
Then, the \emph{cohomological obstruction} of $r_0$ is the element $\obst(r_0)$ of $\CohomRel{1}{C_0}$,
where $\fdec{\gamma}{\CohomRes{0}{U}}{\CohomRel{1}{U}}$ %
is the connecting homomorphism.
\end{definition}
The following proposition justifies regarding these as obstructions.
\begin{proposition}\label{prop:obstruction-compatiblefamily}
Let the cover $\M$ be connected, $C_0 \in \M$,
and $r_0 \in \FF(C_0)$.
Then, $\obst(r_0) = 0$ if and only if there is a compatible family $\enset{r_C \in \FF(C)}_{C \in \M}$ such that $r_{C_0} = r_0$.
\end{proposition}

For a proof of this proposition, see \cite{abramsky2015contextuality}.

\begin{remark}
Note that our cohomology obstruction lives in the first cohomology group. In ordinary homology, the lower groups capture low-dimensional behaviour, and to capture higher-dimensional behaviour, one must pass to the higher homology groups. The situation is quite different in sheaf cohomology; there, it is standard that it is the first cohomology group which captures obstructions to extending local sections to global ones. Of course, it would also be of interest to find natural uses for the higher cohomology groups.
\end{remark}

\subsection{Witnessing contextuality}\label{ssec:cohom-witnesses}

We now apply these tools to analyse the possibilistic structure of empirical models.
The cohomological obstructions of Definition~\ref{def:obs} would appear to be ideally suited to the problem of identifying contextuality. 
The caveat is that, in order to apply those tools, it is necessary to work over a presheaf of abelian groups,
whereas we are concerned with $\Se$, which is merely a presheaf of sets.
We first consider how to build an abelian group from a set.
\begin{definition}
Given a ring $R$, we define a functor $\fdec{\FR}{\Set}{\RMod}$ to the category of $R$-modules (and thus, in particular, to the category of abelian groups).
For each set $X$, $\FR(X)$ is the set of functions $\fdec{\phi}{X}{R}$ of finite support.
Given a function $\fdec{f}{X}{Y}$,
% we define:
\[ \fdecdef{\FR f}{\FR X}{\FR Y}{\phi}{\LambdaAbs{y}\sum_{f(x) = y} \phi(x)} \Mdot \]
\end{definition}
This
% assignment
is easily seen to be functorial. We regard a function $\phi \in \FR(X)$ as a \emph{formal $R$-linear combination}
of elements of $X$: $\sum_{x \in X} \phi(x) \cdot x$.
%\[ \sum_{x \in X} \phi(x) \cdot x \Mdot\]
There is a natural embedding $x \mapsto 1 \cdot x$ of $X$ into $\FR(X)$, which we shall use implicitly throughout.
In fact,
$\FR(X)$ is the \emph{free $R$-module generated by $X$};
and in particular, $\FZ(X)$ is the \emph{free abelian group generated by $X$}.

Given an empirical model $\Se$ defined on the measurement scenario $\tuple{X,\M,O}$,
we shall work with the (relative) \Cech\ cohomology for the abelian presheaf $\FR \Se$
for some ring $R$.
\begin{definition}
With each local section, $s \in \Se(C)$, in the support of an empirical model,
we associate the \emph{cohomological obstruction} $\gamma_{\FR \Se}(s)$.
\begin{itemize}
\item
If there exists some local section $s_0 \in \Se(C_0)$ such that $\gamma_{\FR \Se}(s_0) \neq 0$, we say that $\Se$ is \emph{cohomologically logically contextual}, or $\CLC_R(\Se)$. We also use the more specific notation $\CLC_R(\Se, s_0)$.
\item
If $\gamma_{\FR \Se}(s) \neq 0$ for all local sections, we say that $e$ is \emph{cohomologically strongly contextual}, or $\CSC_R$.
\end{itemize}
\end{definition}
The following
% proposition
justifies considering cohomological obstructions as witnessing contextuality.
\begin{proposition}[{\hspace{1sp}\cite[Proposition 4.3]{abramsky:11a}}]
\label{prop:clcimpsc}
% ~
% \begin{itemize}
% \item
$\CLC_R$ implies $\LC$,
and
%
% \item
$\CSC_R$ implies $\SC$.
%
% \end{itemize}
\end{proposition}
\begin{proof}  %[Proof of Proposition \ref{prop:clcimpsc}]
%{\color{orange}
Suppose an empirical model $e$ is not logically contextual. Then for every maximal context $C_0 \in \M$ and every $s_0 \in \Se(C_0)$,
there is a compatible family $\{ s_C \in \Se(C) \}_{C \in\M}$ with $s_{c_0} = s_0$.
As $\Se(C)$ embeds into $\FR \Se(C)$, $\{s_C\}$ is also a compatible family
in $\FR \Se$. Hence, by Proposition~\ref{prop:obstruction-compatiblefamily}, we conclude that $\obst(s) = 0$. The same argument can be applied to a single section witnessing the failure of strong contextuality.
%}
\end{proof}

Thus we have a sufficient condition for contextuality in the existence of a cohomological obstruction.
Unfortunately, this condition is not, in general, necessary. It is possible that ``false positives'' arise in the form of families $\{ r_C \in \FR \Se(C) \}_{C\in\M}$ which are not
\textit{bona fide} global sections in $\Se(X)$ in which genuine global sections do not exist.

Several examples are discussed in detail in \cite{abramsky:11a}.
%An example for which a false positive arises is the Hardy model \cite{hardy:92,hardy:93}.
It is shown that cohomological obstructions over $\ZZ$ 
provide witnesses of strong contextuality for a number of well-studied models, including:
the GHZ model \cite{greenberger:90},
the Peres--Mermin ``magic'' square \cite{peres:90,mermin:93},
and the 18-vector Kochen--Specker model \cite{cabello:96}, the PR box \cite{popescu:94}, and the Specker triangle \cite{specker:60,liang:11}.
%In fact, the Kochen--Specker model and the Specker triangle belong to a large class of models, known as $\nGCD$.
%
%
%In \cite{abramsky:11a}, it is shown all the models in this class admit cohomological
%witnesses for their strong contextuality. 
These results will be subsumed and greatly generalised in \S\ref{sec:cohomAvN}. 

The coefficients for cohomology can be taken from any commutative ring $R$.
% We now consider
Here is
how the cohomological obstructions obtained with different rings relate to each other:
 \begin{proposition}\label{prop:ringimp}
 Let $\fdec{h}{R'}{R}$ be a ring homomorphism. Then, for any $C \in \M$ and $s \in \Se(C)$,
 $\gamma_{\FRD \Se}(s) = 0$ implies
 $\gamma_{\FR \Se}(s) = 0$,
 and so
 $\CSC_R \implies \CSC_{R'}$ and $\CLC_{R} \implies \CLC_{R'}$.
 \end{proposition}
 \begin{proof} %[Proof of Proposition \ref{prop:ringimp}]
%{\color{orange}
 If $\fdec{h}{R}{R'}$ is a ring homomorphism,
then for any set $X$ there is a map 
\[\fdecdef{F_h}{\FR X}{\FRD X}{r}{h \circ r} \Mdot \]
which is a group homomorphism. Moreover, this assignment is natural in $X$.
Hence, this determines a presheaf map $\FR \Se \longrightarrow \FRD \Se$,
and compatible families on the former are mapped to compatible families on the latter.
Since the map $F_h$ above leaves the elements of the generating set fixed, the condition that the family agrees with $s_0$ at context $C_0$ is preserved. 
%
%
%}
\end{proof}

We conclude this section with a remark.
If $\{ r_C \in \FR \Se(C)\}_{C \in \M}$ is a compatible family, then the sum of the coefficients of the formal linear combinations $r_C$ is the same for all $C$.
This holds because $\Se(\vn) = \EventShf(\vn) = \{\star\}$; so that
for any $C \in \M$, we have
\[
r_C|_\vn (\star) = \sum_{s \in \Se(C)} r_C(s) \Msemicolon
\]
\ie compatibility forces all these restrictions to the empty context to be the same.
Therefore, when the obstruction of a section $s_0 \in \Se(C_0)$
(more precisely, of the linear combination $1 \cdot s_0 \in \FR \Se(C_0)$) vanishes,
the corresponding family of linear combinations $\{ r_C \in \FR \Se(C) \}_{C\in\M}$ must  in fact contain only  \emph{affine} combinations---%
those whose coefficients sum to one.

\section{Cohomology and $\AvN$ arguments}\label{sec:cohomAvN}
The aim of this section is to show that if an empirical model is $\AvN_R$,
then the cohomological obstructions witness its strong contextuality.
Moreover,
it is enough to consider cohomology with coefficients in the ring $R$ itself. 

The result is stated as follows.
\begin{theorem}
\label{thm:main}
  Let $\Se$ be an empirical model on $\tuple{X,\M,R}$. Then:
\[ \AvN_R(\Se) \;\implies\; \SC(\Aff \Se) \;\implies\; \CSC_R(\Se) \;\implies\; \CSC_\ZZ(\Se) \;\implies\; \SC(\Se) 
\Mdot\]  
\end{theorem}

The first of the implications was already established in Proposition~\ref{prop:AvN->SCAff}, the third in Proposition~\ref{prop:ringimp},
and the fourth in Proposition~\ref{prop:clcimpsc}.

In order to prove the second, we use the properties of the functor
$\fdec{\FR}{\Set}{\RMod}$ that
constructs the $R$-module of formal $R$-linear combinations of elements of a set $X$.
As already mentioned,  $\FR X$ is the free $R$-module generated by $X$.
This means that it is the left adjoint of the forgetful functor $\fdec{U}{\RMod}{\Set}$.
\[
\xymatrix@!C=.8cm{
\Set \ar@/^1pc/[rr]^{\FR} &\perp& \RMod \ar@/^1pc/^{U}[ll]
}
\]
The unit $\eta$ of this adjunction is the obvious embedding, which we have been using,
taking an element $x \in X$ to the formal linear combination $1 \cdot x$.
The counit is the natural transformation $\natfdec{\epsilon}{\FR \circ U}{Id_{\RMod}}$
given, for each $R$-module $M$, by the evaluation map
\[\fdecdef{\epsilon_M}{\FR U (M)}{M}{r}{\sum_{x \in M} r(x) x} \Mdot\]

We are interested in taking formal linear combinations of subsets of elements.
Let us fix a module $M$ and a subset $S \subseteq U(M)$.
Then the map $\epsilon_M$, by virtue of being an $R$-module homomorphism,
maps the formal linear combinations of elements of $S$, $\FR(S)$,
which coincide with the linear span in $\FR U (M)$ of $\eta[S] = \setdef{1 \cdot s}{s \in S}$,
to the linear span of $S$ in $M$, $\lspanU{M} S$.
Moreover, it maps the formal affine combinations $\FR^{\mathsf{aff}} (S) = \affU{\FR U (M)} \eta[S]$
to the affine closure $\affU{M} S$.

Recall that we are dealing with measurement scenarios whose outcomes are identified with a ring $R$,
hence where $\EventShf(U)$ are themselves $R$-modules, \ie $\fdec{\EE}{\ps(X)^\op}{\RMod}$.
As such, the counit can be horizontally composed to yield a natural transformation, or map of presheaves,
$\fdec{\mathsf{id}_{\EE} \ast \epsilon}{F_R \circ U \circ \EE}{\EE}$, given at each context $U \subseteq X$ 
by $\fdec{\epsilon_{\EE(U)}}{\FR U \EE(U)}{\EE(U)}$.
Now, given an empirical model $\Se$,
we can apply  the observation regarding subsets of the module at each context.
But, since $\affU{\EE(U)} \Se(U) = \SAffe(U)$ by definition for $U$ beneath the cover, and since containment still holds above it,
we conclude that the presheaf map restricts as follows:
\[
\xymatrix{
  \FR^{\mathsf{aff}} U \Se \ar[d] \ar@{ >->}[r] & \FR U \Se \ar[d] \ar@{ >->}[r] & \FR U \EventShf \ar[d]^\epsilon 
\\
  \Aff \Se \ar@{ >->}[r] & \LSpan \Se \ar@{ >->}[r] & \EventShf 
}
\]
%We can now prove the theorem.
%
%
\begin{proof}[Proof of Theorem~\ref{thm:main}]
We show the contrapositive.
Suppose that $\Se$ is not $\CSC_R$, \ie that $\obst_{\FR \Se}(s_0) = 0$ for some $s_0 \in \Se(C_0)$.
Then, by Proposition~\ref{prop:obstruction-compatiblefamily},
this is equivalent to the existence
of a compatible family $\enset{r_C \in \FR \Se(C)}_{C \in \M}$ with $r_{C_0} = s_0$.
As observed at the end of \S\ref{ssec:cohom-witnesses},
all these $r_C$ must be formal affine combinations of elements in $\Se(C)$.
But then the presheaf map $\FR^{\mathsf{aff}} U \Se  \longrightarrow \Aff\Se$ above
pushes this compatible family to a compatible family of $\Aff\Se$,
implying that the model $\Aff \Se$ is not strongly contextual.
\end{proof}

Essentially the same strategy can be used to prove an analogous result for logical contextuality.
The notion of inconsistent theory has to be adapted: instead of asking
whether there is a global assignment satisfying all the equations in the theory,
we can ask, given a partial assignment $s_0 \in \EE(C_0)$ whether there is such a global assignment
with the additional requirement that it restricts to $s_0$.
This can be seen as a generalisation of the notion of \emph{robust constraint satisfaction} studied in  \cite{abramsky:13b} from the complexity perspective.
We write $\AvN_R(e,s_0)$ if the theory of $\Se$ has no solution extending $s_0$. Then we have:
\[ \AvN_R(e,s_0) \;\implies\; \LC(\Aff \Se,s_0) \;\implies\; \CLC_R(\Se, s_0) 
\;\implies\; \CLC_\ZZ(\Se,s_0) \;\implies\; \LC(\Se, s_0)
\Mdot\]  

Our results show that, where there is an cohomological obstruction, it witnesses genuine contextuality. On the other hand, the important class of $\AvN$ examples are all captured by cohomology. It is worth emphasising that all known quantum examples of strong contextuality are of $\AvN$ type, hence all such examples are captured by cohomology.

\subsubsection*{Discussion}
We have shown that for a large class of models, their logical or strong contextuality is witnessed by cohomology.
This subsumes and greatly generalises the results in \cite{abramsky:11a}.
Moreover, these models include a large class of concrete constructions arising from stabiliser quantum mechanics, going well beyond existing results of this kind in the quantum information literature.
It remains an objective for future work to achieve a precise characterisation of what cohomology detects, and more generally full equivalences between the various ways of expressing contextuality.
Note that, as already mentioned, the first implication in Theorem~\ref{thm:main} can be reversed under the assumption that $R$ is a field.
If we use a more abstract notion of equational consistency, in terms of quotient modules rather than equations expressed in a  ``coordinatized'' form, then it can be reversed even for general rings. The point of taking the ground ring to be a field is exactly that it allows coordinatization.

We also remark that the cohomological methods we have developed can be applied to an elaborated version of the treatment of logical paradoxes we gave in \S\ref{sec:manyfaces}, following the lines of \cite{cook:04,walicki:09}. We aim to give a detailed treatment of this ``cohomology of paradox'' in future work.
We also note the intriguing resemblances, on the conceptual level at least, to the work of Roger Penrose in \cite{penrose1992cohomology}.
\paragraph*{Acknowledgements}
We thank Alexandru Baltag, Louis Narens and Nicholas Teh for useful discussions.
Support from the following is gratefully acknowledged: Templeton World Charity Foundation, AFOSR, EPSRC, the Oxford Martin School, and FCT -- Funda\c{c}\~ao para a Ci\^encia e Tecnologia
(Portuguese Foundation for Science and Technology), 
PhD grant SFRH/BD/94945/2013.
\bibliographystyle{abbrv} 
{\tiny
\bibliography{ccprefs}
}

\end{document}

%% file: pictures/bitregs.tex
\begin{tikzpicture}[scale=2.54]
% dpic version 2011.03.17 option -g for TikZ and PGF 1.01
\ifx\dpiclw\undefined\newdimen\dpiclw\fi
\global\def\dpicdraw{\draw[line width=\dpiclw]}
\global\def\dpicstop{;}
\dpiclw=0.8bp
\dpicdraw[fill=green!15](0,-0.1125) rectangle (0.225,0.1125)\dpicstop
\draw (0.1125,0) node{$0/1$};
\dpicdraw[fill=blue!15](-0.225,-0.5625) rectangle (0,-0.3375)\dpicstop
\draw (-0.1125,-0.45) node{$a_1$};
\dpicdraw[fill=blue!15](0.225,-0.5625) rectangle (0.45,-0.3375)\dpicstop
\draw (0.3375,-0.45) node{$a_2$};
\dpicdraw (-0.1125,-0.3375)
 --(-0.1125,-0.225)
 --(0.3375,-0.225)
 --(0.3375,-0.3375)\dpicstop
\dpicdraw (0.1125,-0.1125)
 --(0.1125,-0.225)\dpicstop
\draw (0,0) node[left=-0.3375bp]{Alice};
\dpicdraw[fill=green!15](1.125,-0.1125) rectangle (1.35,0.1125)\dpicstop
\draw (1.2375,0) node{$0/1$};
\dpicdraw[fill=blue!15](0.9,-0.5625) rectangle (1.125,-0.3375)\dpicstop
\draw (1.0125,-0.45) node{$b_1$};
\dpicdraw[fill=blue!15](1.35,-0.5625) rectangle (1.575,-0.3375)\dpicstop
\draw (1.4625,-0.45) node{$b_2$};
\dpicdraw (1.0125,-0.3375)
 --(1.0125,-0.225)
 --(1.4625,-0.225)
 --(1.4625,-0.3375)\dpicstop
\dpicdraw (1.2375,-0.1125)
 --(1.2375,-0.225)\dpicstop
\draw (1.35,0) node[right=-0.3375bp]{Bob};
\draw (0.675,1.1475) node[above=-0.3375bp]{Target};
\dpicdraw[dashed](0.1125,0.1125)
 --(0.63745,1.069635)\dpicstop
\filldraw[line width=0bp](0.657178,1.058815)
 --(0.65909,1.10909)
 --(0.617723,1.080455) --cycle
\dpicstop
\draw (0.385123,0.60957) node[left=-0.3375bp]{$a_2 = 1$};
\dpicdraw[dashed](1.2375,0.1125)
 --(0.71255,1.069635)\dpicstop
\filldraw[line width=0bp](0.732277,1.080455)
 --(0.69091,1.10909)
 --(0.692822,1.058815) --cycle
\dpicstop
\draw (0.964877,0.60957) node[right=-0.3375bp]{$b_1 = 0$};
\end{tikzpicture}

%% file: pictures/sbitregs.tex
\begin{tikzpicture}[scale=2.54]
% dpic version 2011.03.17 option -g for TikZ and PGF 1.01
\ifx\dpiclw\undefined\newdimen\dpiclw\fi
\global\def\dpicdraw{\draw[line width=\dpiclw]}
\global\def\dpicstop{;}
\dpiclw=0.8bp
\dpicdraw[fill=blue!15](-0.171319,-0.475887) rectangle (0.019035,-0.285532)\dpicstop
\draw (-0.076142,-0.380709) node{$a_1$};
\dpicdraw[fill=blue!15](0.20939,-0.475887) rectangle (0.399745,-0.285532)\dpicstop
\draw (0.304567,-0.380709) node{$a_2$};
\dpicdraw[fill=blue!15](0.818525,-0.475887) rectangle (1.00888,-0.285532)\dpicstop
\draw (0.913702,-0.380709) node{$b_1$};
\dpicdraw[fill=blue!15](1.199234,-0.475887) rectangle (1.389589,-0.285532)\dpicstop
\draw (1.294412,-0.380709) node{$b_2$};
\dpicdraw[fill=lightgray](0.41878,-1.541873) rectangle (0.513958,-1.351518)\dpicstop
\draw (0.466369,-1.446695) node{$0$};
\dpicdraw[fill=lightgray](0.513958,-1.541873) rectangle (0.609135,-1.351518)\dpicstop
\draw (0.561546,-1.446695) node{$1$};
\dpicdraw[fill=lightgray](0.609135,-1.541873) rectangle (0.704312,-1.351518)\dpicstop
\draw (0.656723,-1.446695) node{$0$};
\dpicdraw[fill=lightgray](0.704312,-1.541873) rectangle (0.799489,-1.351518)\dpicstop
\draw (0.751901,-1.446695) node{$1$};
\dpicdraw[fill=lightgray](0.41878,-1.922582) rectangle (0.799489,-1.541873)\dpicstop
\draw (0.609135,-1.732227) node{$\vdots$};
\dpicdraw[dashed](0.466369,-1.351518)
 --(-0.056091,-0.508249)\dpicstop
\filldraw[line width=0bp](-0.039909,-0.498224)
 --(-0.076142,-0.475887)
 --(-0.072272,-0.518275) --cycle
\dpicstop
\dpicdraw[dashed](0.561546,-1.351518)
 --(0.315288,-0.512417)\dpicstop
\filldraw[line width=0bp](0.333553,-0.507056)
 --(0.304567,-0.475887)
 --(0.297023,-0.517777) --cycle
\dpicstop
\dpicdraw[dashed](0.656723,-1.351518)
 --(0.902981,-0.512417)\dpicstop
\filldraw[line width=0bp](0.921247,-0.517777)
 --(0.913702,-0.475887)
 --(0.884716,-0.507056) --cycle
\dpicstop
\dpicdraw[dashed](0.751901,-1.351518)
 --(1.274361,-0.508249)\dpicstop
\filldraw[line width=0bp](1.290542,-0.518275)
 --(1.294412,-0.475887)
 --(1.258179,-0.498224) --cycle
\dpicstop
\draw (0.609135,-1.922582) node[below=-0.285532bp]{Source};
\end{tikzpicture}

%% file: ccp-camera-final.bbl
\begin{thebibliography}{10}

\bibitem{abramsky:12c}
S.~Abramsky.
\newblock Relational databases and {B}ell's theorem.
\newblock In V.~Tannen, L.~Wong, L.~Libkin, W.~Fan, W.-C. Tan, and M.~Fourman,
  editors, {\em In search of elegance in the theory and practice of
  computation}, volume 8000 of {\em LNCS}, pages 13--35. Springer, 2013.

\bibitem{abramsky2015contextuality}
S.~Abramsky, R.~S. Barbosa, K.~Kishida, R.~Lal, and S.~Mansfield.
\newblock Contextuality, cohomology and paradox.
\newblock arXiv preprint arXiv:1502.03097, 2015.

\bibitem{abramsky:11}
S.~Abramsky and A.~Brandenburger.
\newblock The sheaf-theoretic structure of non-locality and contextuality.
\newblock {\em New J.\ Phys.}, 13(11):113036, 2011.

\bibitem{abramsky:13b}
S.~Abramsky, G.~Gottlob, and P.~G. Kolaitis.
\newblock Robust constraint satisfaction and local hidden variables in quantum
  mechanics.
\newblock In F.~Rossi, editor, {\em Proceedings of the Twenty-Third IJCAI},
  pages 440--446. AAAI Press, 2013.

\bibitem{abramsky:11a}
S.~Abramsky, S.~Mansfield, and R.~Soares~Barbosa.
\newblock The cohomology of non-locality and contextuality.
\newblock In B.~Jacobs, P.~Selinger, and B.~Spitters, editors, {\em Proc. 8th
  International Workshop on Quantum Physics and Logic 2011}, volume~95 of {\em
  EPTCS}, pages 1--14, 2012.

\bibitem{bell:64}
J.~S. Bell.
\newblock On the {E}instein-{P}odolsky-{R}osen paradox.
\newblock {\em Physics}, 1(3):195--200, 1964.

\bibitem{cabello:96}
A.~Cabello, J.~M. Estebaranz, and G.~Garc{\'\i}a-Alcaine.
\newblock {B}ell-{K}ochen-{S}pecker theorem.
\newblock {\em Phys.\ Lett.\ A}, 212(4):183--187, 1996.

\bibitem{cabello2010non}
A.~Cabello, S.~Severini, and A.~Winter.
\newblock {(Non-)Contextuality of Physical Theories as an Axiom}.
\newblock {\em arXiv:1010.2163}, 2010.

\bibitem{Caves06}
C.~Caves.
\newblock Stabilizer formalism for qubits.
\newblock Available at \\
  \texttt{info.phys.unm.edu/$\sim$caves/reports/stabilizer.ps}, 2006.

\bibitem{cook:04}
R.~T. Cook.
\newblock Patterns of paradox.
\newblock {\em The Journal of Symbolic Logic}, 69(03):767--774, 2004.

\bibitem{fagin:82}
R.~Fagin, A.~O. Mendelzon, and J.~D. Ullman.
\newblock A simplified universal relation assumption and its properties.
\newblock {\em ACM Transactions on Database Systems (TODS)}, 7(3):343--360,
  1982.

\bibitem{greenberger:90}
D.~M. Greenberger, M.~A. Horne, A.~Shimony, and A.~Zeilinger.
\newblock {B}ell's theorem without inequalities.
\newblock {\em Am.\ J.\ Phys.}, 58(12):1131--1143, 1990.

\bibitem{hardy:93}
L.~Hardy.
\newblock Nonlocality for two particles without inequalities for almost all
  entangled states.
\newblock {\em Phys.\ Rev.\ Lett.}, 71(11):1665--1668, 1993.

\bibitem{howard:14}
M.~Howard, J.~Wallman, V.~Veitch, and J.~Emerson.
\newblock Contextuality supplies the {`}magic{'} for quantum computation.
\newblock {\em Nature}, 510(7505):351--355, 06 2014.

\bibitem{lal:11}
R.~Lal.
\newblock A sheaf-theoretic approach to cluster states, 2011.
\newblock Private communication.

\bibitem{liang:11}
Y.-C. Liang, R.~W. Spekkens, and H.~M. Wiseman.
\newblock {S}pecker's parable of the overprotective seer.
\newblock {\em Phys.\ Rep.}, 506(1):1--39, 2011.

\bibitem{maclane:92}
S.~Mac~Lane and I.~Moerdijk.
\newblock {\em Sheaves in geometry and logic}.
\newblock Springer, 1992.

\bibitem{maier:84}
D.~Maier, J.~D. Ullman, and M.~Y. Vardi.
\newblock On the foundations of the universal relation model.
\newblock {\em ACM Transactions on Database Systems (TODS)}, 9(2):283--308,
  1984.

\bibitem{mansfield:14}
S.~Mansfield.
\newblock Completeness of {H}ardy non-locality: Consequences \& applications.
\newblock In {\em Informal Proceedings of 11th International Workshop on
  Quantum Physics \& Logic}, 2014.

\bibitem{mermin:90b}
N.~D. Mermin.
\newblock Extreme quantum entanglement in a superposition of macroscopically
  distinct states.
\newblock {\em Phys.\ Rev.\ Lett.}, 65(15):1838--1840, 1990.

\bibitem{mermin:90}
N.~D. Mermin.
\newblock Simple unified form for the major no-hidden-variables theorems.
\newblock {\em Phys.\ Rev.\ Lett.}, 65(27):3373--3376, 1990.

\bibitem{mermin:93}
N.~D. Mermin.
\newblock Hidden variables and the two theorems of {J}ohn {B}ell.
\newblock {\em Rev.\ Mod.\ Phys.}, 65(3):803--815, 1993.

\bibitem{nielsen:00}
M.~Nielsen and I.~Chuang.
\newblock {\em Quantum computation and quantum information}.
\newblock Cambridge University Press, 2000.

\bibitem{penrose1992cohomology}
R.~Penrose.
\newblock On the cohomology of impossible figures.
\newblock {\em Leonardo}, 25(3/4):245--247, 1992.

\bibitem{peres:90}
A.~Peres.
\newblock Incompatible results of quantum measurements.
\newblock {\em Phys.\ Lett.\ A}, 151(3-4):107--108, 1990.

\bibitem{pironio:11}
S.~Pironio, J.-D. Bancal, and V.~Scarani.
\newblock Extremal correlations of the tripartite no-signaling polytope.
\newblock {\em J.\ Phys.\ A--Math.\ Theor.}, 44(6):065303, 2011.

\bibitem{popescu:94}
S.~Popescu and D.~Rohrlich.
\newblock Quantum nonlocality as an axiom.
\newblock {\em Found.\ Phys.}, 24(3):379--385, 1994.

\bibitem{raussendorf:01}
R.~Raussendorf and H.~J. Briegel.
\newblock A one-way quantum computer.
\newblock {\em Phys.\ Rev.\ Lett.}, 86(22):5188, 2001.

\bibitem{specker:60}
E.~Specker.
\newblock Die {L}ogik nicht gleichzeitig entscheidbarer {A}ussagen.
\newblock {\em Dialectica}, 14:239--246, 1960.

\bibitem{walicki:09}
M.~Walicki.
\newblock Reference, paradoxes and truth.
\newblock {\em Synthese}, 171(1):195--226, 2009.

\bibitem{wen:01}
L.~Wen.
\newblock Semantic paradoxes as equations.
\newblock {\em Math.\ Intell.}, 23(1):43--48, 2001.

\bibitem{zhang:13}
X.~Zhang, M.~Um, J.~Zhang, S.~An, Y.~Wang, D.-l. Deng, C.~Shen, L.-M. Duan, and
  K.~Kim.
\newblock State-independent experimental test of quantum contextuality with a
  single trapped ion.
\newblock {\em Phys.\ Rev.\ Lett.}, 110(7):070401, 2013.

\bibitem{zu:12}
C.~Zu, Y.-X. Wang, D.-L. Deng, X.-Y. Chang, K.~Liu, P.-Y. Hou, H.-X. Yang, and
  L.-M. Duan.
\newblock State-independent experimental test of quantum contextuality in an
  indivisible system.
\newblock {\em Phys.\ Rev.\ Lett.}, 109(15):150401, 2012.

\end{thebibliography}
